\documentclass{article}

     \PassOptionsToPackage{numbers, compress}{natbib}

   \usepackage[preprint]{neurips_2020}

\usepackage[utf8]{inputenc} 
\usepackage[T1]{fontenc}    
\usepackage{hyperref}       
\usepackage{url}            
\usepackage{booktabs}       
\usepackage{amsfonts}       
\usepackage{nicefrac}       
\usepackage{microtype}      

\usepackage{microtype}
\usepackage{graphicx}
\usepackage{subfigure}
\usepackage{booktabs} 

\usepackage{hyperref}

\usepackage{graphicx} 
\usepackage{subfigure}

\usepackage{algorithm}
\usepackage{algorithmic}

\usepackage{amsmath}
\usepackage{amssymb}

\usepackage{amsthm}
\usepackage{wrapfig}
\usepackage{multicol}

\usepackage{hyperref}

\usepackage{xr}

\makeatletter
\newcommand{\rmnum}[1]{\romannumeral #1}
\newcommand{\Rmnum}[1]{\expandafter\@slowromancap\romannumeral #1@}
\makeatother
\newtheorem{theorem}{\textbf{Theorem}}
\newtheorem{lemma}{\textbf{Lemma}}

\newtheorem{definition}{\textbf{Definition}}

\newtheorem{remark}{\textbf{Remark}}

\newtheorem{claim}{Claim}

\title{A Data-Dependent Algorithm for Querying Earth Mover's Distance with Low Doubling Dimensions}

\author{%
  Hu Ding, Tan Chen, Fan Yang, Mingyue Wang \\
  School of Computer Science and Technology\\
  University of Science and Technology of China\\
  He Fei, China \\
  \texttt{huding@ustc.edu.cn, \{ct1997, yang208, mywang\}@mail.ustc.edu.cn} \\
}

\begin{document}

\maketitle

\begin{abstract}
In this paper, we consider the following query problem: given two weighted point sets $A$ and $B$ in the Euclidean space $\mathbb{R}^d$, we want to quickly determine that whether their earth mover's distance (EMD) is larger or smaller than a pre-specified threshold $T\geq 0$. The problem finds a number of important applications in the fields of machine learning and data mining. In particular, we assume that the dimensionality $d$ is not fixed and the sizes $|A|$ and $|B|$ are large. Therefore, most of existing EMD algorithms are not quite efficient to solve this problem due to their high complexities. Here, we consider the problem under the assumption that $A$ and $B$ have low doubling dimensions, which is common for  high-dimensional data in real world. Inspired by the geometric method {\em net tree}, we propose a novel ``data-dependent'' algorithm to avoid directly computing the EMD between $A$ and $B$, so as to solve this query problem more efficiently. We also study the performance of our method on synthetic and real datasets. The experimental results suggest that our method can save a large amount of running time comparing with existing EMD algorithms. 
  
\end{abstract}

\section{Introduction}
\label{sec-intro}

Given two weighted point sets $A$ and $B$ in the Euclidean space $\mathbb{R}^d$, the weight of each point in $A$ represents its {\em supply},  and the weight of each point in $B$ represents its {\em demand}. The {\em Earth Mover's Distance (EMD)}  is the minimum transportation cost from $A$ to $B$. We can build a complete bipartite graph $A\times B$ where each pair of points $(a_i, b_j)\in A\times B$ is connected by an edge with the weight being equal to their Euclidean distance (or some other specified distance); so the EMD between $A$ and $B$ can be obtained by computing the {\em minimum cost flow}~\cite{ahuja1993network} in the bipartite graph. Actually, the transportation problem is a discrete version of the {\em Monge-Kantorovich} problem that has been extensively studied in Mathematics~\cite{Villani08references}. EMD has many applications in real world. In particular, it has been widely used for computing the similarity between two patterns in pattern recognition and image retrieval problems~\cite{rubner2000earth,DBLP:conf/cvpr/GraumanD04,pele2009fast}. Most of these applications consider the EMD in terms of low dimensional patterns, such as 2D images and 3D shapes. In recent years, several important applications of EMD in high-dimensional Euclidean space have been studied in the fields of machine learning and data mining. Below, we introduce two examples briefly. 

\textbf{\rmnum{1}. Computing similarity between different datasets.} 
{\em Crowdsourcing} is an emerging topic in the big data era~\cite{DBLP:journals/sigkdd/LiGMLSZFH15,ZCZ}. We often receive different datasets from various sources and want to quickly estimate their values. For example, we can perform evaluation or classification on the received datasets by comparing them with our own reliable datasets. In practice, a dataset ({\em e.g.,} an image dataset) is often represented as a set of high-dimensional vectors, and therefore the task can be modeled as  computing the similarity  between two point sets in a \textbf{high-dimensional} Euclidean space (this is the key difference with the problem of computing the EMD between two images, which are in fact two \textbf{2D point sets}~\cite{cohen1999earth}). Similar applications also arise in biological network alignment~\cite{DBLP:conf/aaai/LiuDC017} and unsupervised cross-lingual learning~\cite{DBLP:conf/emnlp/ZhangLLS17}.

\textbf{\rmnum{2}. Domain adaptation.} 
In supervised learning, our task usually is to learn the knowledge from a given labeled training dataset. However, in many scenarios, labeled data could be very limited. We can generate the labels for an unlabeled dataset by exploiting an existing annotated dataset, that is, transfer the knowledge from a source domain to a target domain. Thus the problem is called ``domain adaptation'' in the field of transfer learning~\cite{DBLP:journals/tkde/PanY10}. Due to its importance to many machine learning applications, the problem has received a great amount of attention in the past years~\cite{DBLP:conf/nips/BlitzerCKPW07,DBLP:journals/ml/Ben-DavidBCKPV10}. Recently,  Courty {\em et al.} \cite{DBLP:journals/pami/CourtyFTR17} modeled the domain adaptation problem as a transportation problem of computing the EMD between the source and target domains. Similar to the first application, the datasets are often represented as high-dimensional point sets and therefore we need to compute their EMD in high dimension. 


In many practical scenarios, we usually only need to \textbf{quickly answer the question that whether the EMD between the given point sets is larger or smaller than a threshold}, instead of returning the exact EMD value or the EMD induced map in the bipartite graph $A\times B$. For example, given two large-scale datasets, we may just want to know that whether they are similar enough and do not care about the detailed map of the data items; for the domain adaptation problem, we may want to quickly determine that whether the given annotated dataset is a suitable source for the unlabeled dataset before conducting the expensive computation for the transportation problem. Thus, it is critical to design a fast algorithm to satisfy these applications. However, existing methods for computing EMD or estimating EMD bounds often suffer from the issues like high complexity or high distortion in high dimensions (a detailed discussion on existing methods is given in Section~\ref{sec-relate}). 

In this paper, we study the \textbf{EMD query problem}: given a value $T\geq 0$, we want to quickly determine that whether the EMD between $A$ and $B$ is larger or smaller than it.  In particular we consider the data having low doubling dimension (we will provide the formal definition of doubling dimension in Section~\ref{sec-pre}). The doubling dimension is a measure that has been widely adopted in machine learning community for describing the intrinsic dimensionality of data~\cite{DBLP:journals/jcss/BshoutyLL09,DBLP:journals/talg/ChanGMZ16}. Note that  real-world high-dimensional data often reveals small intrinsic dimension. For example, an image set can be represented as a set of vey high-dimensional vectors, while the  vectors may be distributed nearby a low-dimensional manifold; thus their intrinsic dimension could be much smaller than the dimension of the Euclidean space~\cite{belkin2003problems}. 


\textbf{Our contribution.} 
In this paper, we develop a ``data-dependent'' algorithm for solving the EMD query problem. Our algorithm relies on a hierarchical structure that can be viewed as a simplified {\em net tree} in doubling metrics~\cite{har2006fast}. In particular, the height of the structure depends on how close the exact EMD and $T$ are. Specifically, the lager the difference between the two values, the lower the structure (and the lower the running time). 
Besides the low complexity, our algorithm also enjoys the following two advantages. \textbf{(1).} Our algorithm does not need to build any complicated data structure and is easy to implement in practice. Also, our method can be easily modified to handle the case that the doubling dimension of the data is not given in advance. \textbf{(2).} Our algorithm actually is a general framework for solving the EMD query problem, that is, any existing EMD algorithm can be plugged as the black box for computing the EMD of an ``easy'' instance in each round of our framework. Hence the efficiency of our framework can be always improved if any new EMD algorithm is proposed in future.

\vspace{-0.1in}
\subsection{Existing Methods for Computing EMD}
\label{sec-relate}
\vspace{-0.1in}

A number of minimum cost flow algorithms have been developed in the past decades~\cite{ahuja1993network,orlin1993polynomial,orlin1997polynomial,DBLP:journals/combinatorica/Tardos85,DBLP:journals/jacm/GoldbergT89}.
Suppose $n$ and $m$ are the numbers of vertices and edges in the bipartite graph respectively, and $U$ is the maximum weight. 
Orlin~\cite{DBLP:conf/stoc/Orlin88} developed a strongly polynomial algorithm with the time complexity $O(n\log n (m+n\log n))$.
 Lee and  Sidford~\cite{DBLP:conf/focs/LeeS14} designed a novel linear solver and one can apply it to solve the minimum cost flow problem in $O(n^{2.5} poly(\log U))$ time. 
Using the idea of preconditioning, Sherman~\cite{DBLP:conf/soda/Sherman17} provided a $(1+\epsilon)$-approximation algorithm with the running time $O(n^{2+o(1)}\epsilon^{-2})$.
For the instances in Euclidean space, a sequence of faster algorithms have been proposed in the community of computational geometry. For example, Khesin {\em et al.}~\cite{DBLP:conf/compgeom/KhesinNP19} applied the idea of preconditioning~\cite{DBLP:conf/soda/Sherman17} to design two randomized $(1+\epsilon)$-approximation nearly linear time algorithms (if the dimension $d$ is a constant number).

Several practical EMD algorithms for low-dimensional patterns, like 2D images, were proposed before~\cite{ling2007efficient,DBLP:conf/cvpr/ShirdhonkarJ08,pele2009fast,DBLP:journals/pvldb/TangUCMC13,Chan2019ThePO}. 
In the community of machine learning, Cuturi~\cite{DBLP:conf/nips/Cuturi13} proposed a new objective called ``Sinkhorn Distance'' that smoothes the transportation problem with an entropic regularization term, and it can be solved much faster than computing the exact EMD; Li {\em et al.}~\cite{DBLP:journals/jscic/LiROYG18} designed a parallel method for computing EMD. Following Cuturi's work, several improved Sinkhorn algorithms have been proposed~\cite{DBLP:conf/nips/AltschulerWR17,DBLP:conf/nips/AltschulerBRN19,DBLP:conf/nips/MuzellecC19} very recently. Kusner {\em et al.}~\cite{kusner2015word} modified the objective of EMD and proposed a new distance called ``Relaxed Word Mover's Distance (RWMD)'' which is easier to compute; Atasu and Mittelholzer~\cite{DBLP:conf/icml/AtasuM19} further showed a linear time parallel RWMD algorithm. 


Several efficient algorithms have been developed for estimating the EMD without computing the induced map between the given point sets. For example, Indyk~\cite{indyk2007near} gave a near linear time constant factor approximation algorithm by using the importance sampling technique; Cabello {\em et al.}~\cite{cabello2008matching} showed that it is possible to achieve a $(1+\epsilon)$-approximation in $O(\frac{n^2}{\epsilon^2}\log^2 n)$ time by constructing the geometric spanner; Andoni {\em et al.}~\cite{DBLP:conf/stoc/AndoniNOY14} gave a streaming algorithm that can return a $(1+\epsilon)$-approximation estimate in $O(n^{1+o(1)})$ time. 
However, most of these algorithms rely on the geometric techniques in low-dimensional space, and their complexities are exponential in the dimensionality $d$. Li~\cite{DBLP:journals/corr/abs-1002-4034} generalized the method of~\cite{indyk2007near} and proposed an $O(\rho)$-approximate estimate of EMD where $\rho$ is the doubling dimension of the given data; however, the algorithm needs a $O(n^2 poly(\log n))$ preprocessing time that could be too high when $n$ is large.

Another natural approach for computing EMD is {\em metric embedding}
~\cite{DBLP:conf/soda/AndoniIK08,IT03}. However, this approach often has a large distortion ({\em e.g.,} $O(\log n\cdot \log d)$ in~\cite{DBLP:conf/soda/AndoniIK08}), and thus is not suitable for solving our problem with large $n$ and $d$.

\subsection{Preliminaries}
\label{sec-pre}

We introduce several important definitions that will be used throughout this paper.

\begin{definition}[Earth Mover's Distance (EMD)]
\label{def-emd}
Let $A=\{a_1, a_2, \cdots, a_{n_A}\}$ and $B=\{b_1, b_2, \cdots, b_{n_B}\}$ be two sets of weighted points in $\mathbb{R}^d$ with nonnegative weights $\alpha_i$ and $\beta_j$ for each $a_i\in A$ and $b_j\in B$, and $\sum^{n_A}_{i=1}\alpha_i=\sum^{n_B}_{j=1}\beta_j=W$. Their earth mover's distance is
\begin{eqnarray}
\mathcal{EMD}(A, B)=\frac{1}{W}\min_{F}\sum^{n_A}_{i=1}\sum^{n_B}_{j=1}f_{ij}||a_i-b_j||, \label{for-emd}
\end{eqnarray} 
where $||\cdot||$ indicates the Euclidean distance and $F=\{f_{ij}\mid 1\leq i\leq n_A, 1\leq j\leq n_B\}$ is a feasible flow from $A$ to $B$, i.e., each $f_{ij}\geq 0$, $\sum^{n_A}_{i=1}f_{ij}=\beta_j$, and $\sum^{n_B}_{j=1}f_{ij}=\alpha_i$. 
\end{definition}

\begin{definition}[EMD Query]
\label{def-query}
Given two weighted point sets $A$ and $B$ in $\mathbb{R}^d$ and  $T\geq 0$, the problem of EMD Query is to answer the question that whether $\mathcal{EMD}(A, B)\geq T$ or  $\mathcal{EMD}(A, B)\leq T$.
\end{definition}

For any point $p\in \mathbb{R}^d$ and $r\geq 0$, we use $Ball(p, r)=\{q\in  \mathbb{R}^d\mid ||q-p||\leq r\}$ to indicate the ball of radius $r$ around $p$. Usually, the doubling dimension is defined for an abstract metric space~\cite{DBLP:journals/talg/ChanGMZ16}. In this paper, since we focus mainly on the applications for high-dimensional data with low intrinsic dimensions, we directly describe the doubling dimension for point sets in Euclidean space.

\begin{definition}[Doubling Dimension]
\label{def-dd}
The doubling dimension of a point set $P\subset\mathbb{R}^d$ is the smallest number $\rho$, such that for any $p\in P$ and $r\geq 0$, $P\cap Ball(p, 2r)$ is always covered by the union of at most $2^\rho$ balls with radius $r$.
\end{definition}
The doubling dimension describes the expansion rate of $P$. For example, imagine a set of points uniformly distributed in a $d'$-dimensional flat in $\mathbb{R}^d$, and then the doubling dimension is $O(d')$ but the Euclidean dimension $d$ can be much higher. 

\begin{claim}
\label{cla-dd}
Let $A$ and $B$ be two point sets in $\mathbb{R}^d$ with each one having the doubling dimension $\rho>0$. Then the set $A\cup B$ has the doubling dimension at most $\rho+1$.
\end{claim}
This claim is easy to verify. Given any ball $Ball(p, 2r)$, we have $\big(A\cup B\big)\cap Ball(p, 2r)=\big(A\cap Ball(p, 2r)\big)\cup \big(B\cap Ball(p, 2r)\big)$. So $\big(A\cup B\big)\cap Ball(p, 2r)$ is covered by at most $2^\rho+2^\rho=2^{\rho+1}$ balls with radius $r$. Therefore, Claim~\ref{cla-dd} is true.

\textbf{The rest of the paper is organized as follows.} In Section~\ref{sec-bc}, we introduce a simplified variant of the net tree method to hierarchically decompose a given set of points in the space.
By using the algorithm proposed in Section~\ref{sec-bc}, we introduce our method for solving the EMD query problem (approximately) in Section~\ref{sec-match}. Finally, we evaluate the experimental  performances in Section~\ref{sec-exp}.

\section{Hierarchical Gonzalez's algorithm}
\label{sec-bc}

In this section, we propose a hierarchical algorithm to decompose the point set from coarse to fine. Roughly speaking, given a set $P$ of $n$ points in $\mathbb{R}^d$, we partition it to be covered by a set of balls where the number of the balls is bounded; then we recursively perform the same strategy for the points inside each individual ball until the radius becomes small enough. It is easy to see that this approach will yield a tree, where each node of the tree corresponds to an individual ball and its children form a decomposition of the points inside the ball. 

The structure actually can be realized by constructing the {\em net tree} which has been particularly studied in the context of doubling metrics. Har-Peled and Mendel~\cite{har2006fast} showed that the net tree can be constructed in $2^{O(\rho)}nd\log n$ expected time if the point set has the doubling dimension $\rho$; their idea is based on a fast implementation of the well-known $k$-center clustering algorithm~\cite{gonzalez1985clustering} and the method of {\em hierarchically well-separated tree (HST)}~\cite{DBLP:conf/focs/Bartal96}. 
However, their method needs to maintain and update some auxiliary data structures that are not very efficient for handling large-scale datasets in practice. Moreover, the method takes an extra $O(2^{O(\rho)}n)$ space for maintaining the data structures (besides the original $O(nd)$ for storing the input data).

\begin{algorithm}[tb]
   \caption{\textsc{Hierarchical Gonzalez's algorithm}}
   \label{alg-hg}
\begin{algorithmic}
  \STATE {\bfseries Input:} A set $P$ of $n$ points in $\mathbb{R}^d$, a parameter $h>0$, and the doubling dimension $\rho$. 
   \STATE
   \begin{enumerate}
   \item Initialize an empty tree $\mathcal{H}$, and each node $v$ of $\mathcal{H}$ is associated with a point $p_v$ and a subset $P_v$ of $P$.  
   \item Arbitrarily select a point $p_0\in P$. Let  the root node of $\mathcal{H}$ be $v_0$. Also, set $p_{v_0}=p_0$ and $P_{v_0}=P$. The root $v_0$ is labeled as the $0$-th level node. 
   \item Starting from $v_0$, recursively grow each node $v$ of $\mathcal{H}$ as follows:
   \begin{enumerate}
  \item Suppose the level of $v$ is $i\geq 0$. If $i=h$ or $P_v$ contains only one point, $v$ will be a leaf and stop growing it.
  \item Else, run the Gonzalez's algorithm $2^{2\rho}$ rounds on $P_v$, and obtain the $2^{2\rho}$ clusters with their cluster centers; add $2^{2\rho}$ children nodes to $v$, where each child is associated with an individual cluster of $P_v$ and the corresponding cluster center. Each child is labeled as a $(i+1)$-level node. 
     \end{enumerate}
   
   \end{enumerate}
 \vskip -0.05in
\end{algorithmic}
\vskip -0.05in
\end{algorithm}

\textbf{Our approach and high-level idea.} 
In a standard net tree, the nodes at the same level are required to satisfy two key properties: the ``covering property'' and ``packing property''. Informally speaking, at each level of the net tree, the covering property requires that each point of $P$ should be covered by a ball centered at one node (each node has a ``representative'' point from $P$) with a specified radius; the packing property requires that the representatives of the nodes are ``well separated'' ({\em i.e.,} their inter distances should be large enough). 
We observe that  the ``packing property'' is not a necessary condition to solve our EMD query problem. Hence our proposed algorithm can be viewed as a simplified variant of the net tree method, which only keeps the covering property and takes only a $O(nd)$ space complexity. 

Our algorithm also relies on the Gonzalez's $k$-center clustering algorithm~\cite{gonzalez1985clustering}, and we briefly introduce it for the sake of completeness. Initially, it selects an arbitrary point, say $c_1$, from the input $P$ and lets $S=\{c_1\}$; then it iteratively selects a new point that has the largest distance to $S$ among the points of $P$ and adds it to $S$, until $|S|=k$ (the distance between a point $q$ and $S$ is defined as $dist(q, S)=\min\{||q-p||\mid p\in S\}$); suppose $S=\{c_1, \cdots, c_k\}$, and then $P$ is covered by the $k$ balls $Ball(c_1, r), \cdots, Ball(c_k, r)$ with  $r\leq\min\{||c_i-c_j||\mid 1\leq i\neq j\leq k\}$. It is easy to know that the running time of the Gonzalez's algorithm is $O(|S|nd)$.

Our main idea is to hierarchically decompose the given point set and run the Gonzalez's algorithm locally, and therefore we name the algorithm as {\em \textsc{Hierarchical Gonzalez's algorithm}} (see Algorithm~\ref{alg-hg}). Denote by $\Delta$ the radius of the minimum enclosing ball of $P$. Initially, the whole point set $P$ is covered by a ball with radius $\Delta$. By applying Definition~\ref{def-dd} twice, we know that $P$ is covered by $2^{2\rho}$ balls $\mathbb{B}=\{\mathcal{B}_1, \cdots, \mathcal{B}_{2^{2\rho}}\}$ with radius $\Delta/4$ (note that we can only claim these balls exist, but cannot find these balls explicitly). If running the Gonzalez's algorithm $2^{2\rho}$ rounds, we obtain $2^{2\rho}$ points, say $\{s_1, s_2, \cdots, s_{2^{2\rho}}\}$, and consider two cases: the points separately fall into different balls of $\mathbb{B}$ or not. For the first case, through the triangle inequality we know that $P$ is covered by $\cup^{2^{2\rho}}_{j=1}Ball(s_j, \Delta/2)$. For the other case ({\em i.e.,} there exist two points, say $s_{j_1}$ and $s_{j_2}$, falling into one ball, and thus the  distance $||s_{j_1}-s_{j_2}||\leq \Delta/2$), due to the nature of the Gonzalez's algorithm, we know that for each point $p\in P$, 
\begin{eqnarray}
\min_{1\leq j\leq 2^{2\rho}}||p-s_j||\leq \min_{1\leq j<j'\leq 2^{2\rho}}||s_j-s_{j'}||\leq ||s_{j_1}-s_{j_2}||\leq \Delta/2. \label{for-tree1}
\end{eqnarray}
Thus, for the second case, $P$ is also covered by  $\cup^{2^{2\rho}}_{j=1}Ball(s_j, \Delta/2)$. Namely, we decompose $P$ into $2^{2\rho}$ parts and each part is covered by a ball with radius $\Delta/2$. In the following steps, we just recursively run the Gonzalez's algorithm on each part locally. If we perform $\log\frac{\Delta}{r}$ rounds with a specified value $r>0$, each point of $P$ will be covered by a ball with radius $r$. Moreover, we can imagine that the algorithm generates a hierarchical tree $\mathcal{H}$ with height $h=\log\frac{\Delta}{r}+1$, where the root ($0$-th level) corresponds to the set $P$ and each node at the $i$-th level, $1\leq i\leq \log\frac{\Delta}{r}$, corresponds to a subset of $P$ that is covered by a ball with radius $\Delta/2^i$. Obviously, each leaf node is covered by a ball with radius $r$, and the total number of leaves is $\min\{n, (2^{2\rho})^{\log\frac{\Delta}{r}}\}=\min\{n,(\frac{\Delta}{r})^{2\rho}\}$ (we stop growing the node if its corresponding subset has only one point). 

\textbf{Running time.} 
For the $i$-th level of $\mathcal{H}$, denote by $n_1, n_2, \cdots, n_{2^{2\rho i}}$ the number of points covered by the $2^{2\rho i}$ nodes, respectively (obviously, $\sum^{2^{2\rho i}}_{j=1}n_j=n$). For each node, we run the Gonzalez's algorithm $2^{2\rho}$ rounds locally. Therefore, the total running time cost at the $i$-th level is 
\begin{eqnarray}
\sum^{2^{2\rho i}}_{j=1}O(2^{2\rho}n_j d)=O(2^{2\rho}n d). \label{for-tree2}
\end{eqnarray}
Consequently, the total running time of Algorithm~\ref{alg-hg} is $O(2^{2\rho} (\log \frac{\Delta}{r})nd)$ if $h=\log \frac{\Delta}{r}+1$.


\textbf{Space complexity.} 
In Section~\ref{sec-match}, we will show that we actually do not need to store the whole $\mathcal{H}$. Instead, we conduct the computation  from top to bottom along $\mathcal{H}$. The space used for the $i$-th level can be released when the nodes at the $(i+1)$-th level all have been generated. That is, we just need to store at most two levels when constructing the tree $\mathcal{H}$ in  Algorithm~\ref{alg-hg}. Also, the space used for storing each level is always $O(nd)$. Therefore, the space complexity of Algorithm~\ref{alg-hg} is $O(nd)$.

Overall, we have the following theorem.

\begin{theorem}
\label{the-tree}
Let $r>0$ be a given number. If we set $h=\log \frac{\Delta}{r}+1$, the \textsc{Hierarchical Gonzalez's algorithm} (Algorithm~\ref{alg-hg}) generates a set of $\min\{n, (\frac{\Delta}{r})^{2\rho}\}$ balls covering $P$ with radius $r$, in $O(2^{2\rho} (\log \frac{\Delta}{r})nd)$ time. The space complexity is $O(nd)$.
\end{theorem}

\begin{remark} [\textbf{If $\rho$ is not given}]
\label{rem-tree}
In Algorithm~\ref{alg-hg}, we require to input the doubling dimension $\rho$. Actually this is not necessary. Assume we know the value of $\Delta$ ({\em i.e.,} the radius of the minimum enclosing ball of $P$). Then, for each node  
$v$ at the $i$-th level of $\mathcal{H}$, we just run the Gonzalez's algorithm $k$ rounds until the obtained $k$ clusters have radius at most $\Delta/2^{i+1}$; by the same manner of our previous analysis, we know that $k$ should be no larger than $2^{2\rho}$. Therefore, we have the same time and space complexities as Theorem~\ref{the-tree}. 

But it is expensive to compute the exact value of $\Delta$. One solution is to compute an approximate minimum enclosing ball of $P$ ({\em e.g.,}  the $O(\frac{1}{\epsilon}nd)$ time $(1+\epsilon)$-approximation algorithm of \cite{badoiu2003smaller}). Actually, we can solve this issue by a much simpler way. We can arbitrarily select a point $p$ and its farthest point $p'$ from $P$ (this step takes only linear time), and it is easy to see that $||p-p'||\in [\Delta, 2\Delta]$; then we just replace $\Delta$ by the value $\tilde{\Delta}=||p-p'||$ in the algorithm. Since $\tilde{\Delta}\leq 2 \Delta$, the height of the tree $\mathcal{H}$ will be at most $\log\frac{\tilde{\Delta}}{r}+1\leq \log\frac{\Delta}{r}+2$ (so we just increase the height by one).

\end{remark}

\vspace{-0.07in}
\section{EMD Query Algorithm}
\label{sec-match}
\vspace{-0.07in}

The recent hardness-of-approximation result reveals that it is quite unlikely to achieve an algorithm being able to solve the EMD query problem with a low time complexity. 
Under the Hitting Set Conjecture, 
Rohatgi~\cite{DBLP:conf/approx/Rohatgi19} proved that there is no truly subquadratic time algorithm yielding an approximate EMD in high dimensions. In this section, we consider solving the EMD query problem in a more efficient way. To better understand our algorithm, we introduce the high-level idea first.

\textbf{High-level idea of Algorithm~\ref{alg-hemd}.} To avoid directly solving the challenging EMD problem, we relax the requirement of Definition~\ref{def-query} slightly. Our intuition is similar to the relaxation for the nearest-neighbor search problem by {\em Locality-Sensitive Hashing}, which distinguishes the cases that the distance is smaller than $R$ or larger than $cR$ for some $R>0$ and $c>1$~\cite{DBLP:journals/cacm/AndoniI08}. Let $(A, B, T)$ be an instance of Definition~\ref{def-query}. Suppose $\epsilon>0$ is a given small parameter, and for simplicity we let $\Delta$ be the maximum radius of the minimum enclosing balls of $A$ and $B$. Our idea is to distinguish the instances ``$\mathcal{EMD}(A, B)> T+ \epsilon \Delta$'' and ``$\mathcal{EMD}(A, B)<T- \epsilon \Delta$''; the term ``$\epsilon \Delta$'' can be viewed as the induced approximation error. To realize this goal, we use the hierarchical structure $\mathcal{H}$ constructed in Algorithm~\ref{alg-hg} to estimate the value of $\mathcal{EMD}(A, B)$ from coarse to fine, until these two instances can be distinguished. At each level, we just need to compute an easy instance, $\mathcal{EMD}(A_i, B_i)$, where the sizes of $A_i$ and $B_i$ are much smaller; then we use the obtained value $\mathcal{EMD}(A_i, B_i)$ to determine whether we need to go deeper (see Step 2(a)-2(d)). 
For ease of presentation, we name 
the following three cases:
  \textbf{case 1:} $\mathcal{EMD}(A, B)> T$;
 \textbf{case 2:} $\mathcal{EMD}(A, B)< T$;
 \textbf{case 3:} $\mathcal{EMD}(A, B)\in T\pm \epsilon \Delta$.

\begin{theorem}
\label{the-hemd}
There are $4$ possible events in total. 
\textbf{(\rmnum{1})} If $\mathcal{EMD}(A, B)> T+ \epsilon \Delta$, Algorithm~\ref{alg-hemd} will return ``case 1''. \textbf{(\rmnum{2})} If $\mathcal{EMD}(A, B)<T- \epsilon \Delta$, the algorithm will return ``case 2''. \textbf{(\rmnum{3})}  If $\mathcal{EMD}(A, B)\in [T, T+ \epsilon \Delta]$, the algorithm will return ``case 1'' or  ``case 3''. \textbf{(\rmnum{4})}  If $\mathcal{EMD}(A, B)\in [ T- \epsilon \Delta, T]$, the algorithm will return ``case 2'' or  ``case 3''. 
The height of the tree $\mathcal{H}$ built in  Algorithm~\ref{alg-hemd}  is at most $\min\{\log\frac{1}{\epsilon}, \log\frac{\Delta}{\delta}\}+5$ where $\delta=\big|\mathcal{EMD}(A, B)-T\big|$.
\end{theorem}

\begin{remark}
\label{rem-the-hemd}
(\rmnum{1}) The algorithm relies on Algorithm~\ref{alg-hg}, and thus it can be also easily modified for solving the case that the doubling dimension is not given (see our analysis in Remark~\ref{rem-tree}). 

(\rmnum{2}) The \textbf{running time} of Algorithm~\ref{alg-hemd}  is data-dependent. The height of the tree $\mathcal{H}$ depends on how close $\mathcal{EMD}(A, B)$ and $T$ are and how accurate we require the solution to be. Specifically, the closer the values or the smaller the error parameter $\epsilon$, the higher the structure (and the higher the running time). As the sub-routine, we can apply any existing EMD algorithm $\mathcal{A}$ to compute $\mathcal{EMD}(A_i, B_i)$ at the $i$-th level of $\mathcal{H}$ (see Step 2(d)). Suppose the time complexity of $\mathcal{A}$ is $\Gamma(n_A, n_B)$ for computing the original instance $(A, B)$. Since the time function $\Gamma(\cdot, \cdot)$ usually is super-linear and the total size $|A_i|+|B_i|$ increases at a geometric rate over $i$, the complexity of Algorithm~\ref{alg-hemd} will be dominated by the running time at the last $h$-th level of $\mathcal{H}$ plus the complexity of constructing $\mathcal{H}$, {\em i.e.,} 
\begin{eqnarray}
\Gamma\big(|A_h|,|B_h|\big)+O\big(2^{2(\rho+1)}\cdot h\cdot (n_A+n_B)\cdot d\big), \label{for-rem-the-hemd-1}
\end{eqnarray}
where $|A_h|$ and $|B_h|$ are at most $2^{2(\rho+1)h}$. The height  $h$  is at most $\min\{\log\frac{1}{\epsilon}, \log\frac{\Delta}{\delta}\}+5$ due to Theorem~\ref{the-hemd}. If $n_A$ and $n_B$ are much larger than $2^{2(\rho+1)h}$, the complexity (\ref{for-rem-the-hemd-1}) is linear in the input size. Note that $\Gamma(n_A, n_B)$ usually is at least $\Omega(n_A\cdot n_B\cdot d)$, and thus our method can save a substantial amount of the running time especially when the data sizes and dimensionality are large.



(\rmnum{3}) The recent work~\cite{Chan2019ThePO} also considered outputting the bounds of EMD. However, it requires the pairwise ground distances to be given, which need $\Omega(n_A\cdot n_B\cdot d)$ time to compute, before computing the EMD. It could be very expensive when the data sizes and $d$ are large. On the other hand, our method avoids this by using the hierarchical structure. Also, the method of~\cite{Chan2019ThePO} is not always guaranteed to output the desired bounds of EMD (the algorithm may fail to generate a tight enough bound), while our method has a strict guarantee of the correctness as Theorem~\ref{the-hemd}. 

\end{remark}

\begin{algorithm}[tb]
   \caption{\textsc{Hierarchical EMD Query Algorithm}}
   \label{alg-hemd}
\begin{algorithmic}
  \STATE {\bfseries Input:} Two point sets $A$ and $B$ in $\mathbb{R}^d$, and the doubling dimension $\rho$. $T>0$ and $\epsilon\in (0,1)$. 
   \STATE
   \begin{enumerate}
   \item Compute the approximate radius of the minimum enclosing balls of $A$ and $B$ (via the method mentioned in Remark~\ref{rem-tree}), and denote them as $\tilde{\Delta}_{A}$ and $\tilde{\Delta}_{B}$ respectively. Let $\tilde{\Delta}=\max\{\tilde{\Delta}_{A},\tilde{\Delta}_{B}\}$. 
     \item Let $P=A\cup B$. Construct the tree $\mathcal{H}$ level by level (from top to bottom) via Algorithm~\ref{alg-hg} (replace $\rho$ by $\rho +1$ according to Claim~\ref{cla-dd}).
     \begin{enumerate}
     \item Let $i$ be the index of the current level at $\mathcal{H}$. If $i=\log \frac{1}{\epsilon}+5$, stop the loop and output ``\textbf{Case 3}''.

           \item Let $v^1_i, v^2_i, \cdots, v^{N}_i$ be the nodes at the $i$-th level ($N=2^{2(\rho+1) i}$). Correspondingly, each node $v^j_i$ is associated with a point $p_{v^j_i}$ and a subset $P_{v^j_i}$ of $P$. Let $n^j_i=$ the total weight of $A\cap P_{v^j_i}$ and $m^j_i=$ the total weight of $B\cap P_{v^j_i}$.  
      \item Initialize two empty sets of points $A_i$ and $B_i$. For each $v^j_i$, $1\leq j\leq N$, if $n^j_i\geq m^j_i$, add $p_{v^j_i}$ to $A_i$ and assign a weight $n^j_i-m^j_i$ to it; else, add $p_{v^j_i}$ to $B_i$ and assign a weight $m^j_i-n^j_i$ to it.

      \item Compute $\mathcal{EMD}(A_i, B_i)$  by an existing EMD algorithm. If $\mathcal{EMD}(A_i, B_i)\geq T+\frac{1}{2^{i-3}}\tilde{\Delta}$, stop the loop and output ``\textbf{Case 1}''; else if $\mathcal{EMD}(A_i, B_i)\leq T-\frac{1}{2^{i-3}}\tilde{\Delta}$, stop the loop and output ``\textbf{Case 2}''.
     \end{enumerate}
        \end{enumerate}
 
\end{algorithmic}
\end{algorithm}

To prove Theorem~\ref{the-hemd}, we need to prove the following Lemma~\ref{lem-phemd} and \ref{lem-hemd} first. 

We let $\Delta_A$ and $\Delta_B$ be the radii of the minimum enclosing balls of $A$ and $B$, respectively, and thus $\Delta=\max\{\Delta_A, \Delta_B\}$. 
Note that the radius of the minimum enclosing ball of $P=A\cup B$ could be much larger than $\Delta$. But Lemma~\ref{lem-phemd} tells us that after the first level, the approximation  error only depends on $\Delta$.

\begin{lemma}
\label{lem-phemd}
(\rmnum{1}) The value $\tilde{\Delta}$ obtained in Step 1 of Algorithm~\ref{alg-hemd} is between $\Delta$ and $2\Delta$. (\rmnum{2}) At the first level of $\mathcal{H}$, the set $P$ is decomposed into $2^{2(\rho+1)}$ balls where each ball has the radius at most $2\Delta$.
\end{lemma}

\begin{proof}
 It is easy to prove the statement (\rmnum{1}). Since $\tilde{\Delta}_{A}\in [\Delta_A, 2\Delta_A]$ and $\tilde{\Delta}_{B}\in [\Delta_B, 2\Delta_B]$, we directly have $\tilde{\Delta}=\max\{\tilde{\Delta}_A, \tilde{\Delta}_B\}\in [\Delta, 2\Delta]$. 

We can view the set $P=A\cup B$ as an instance of $2$-center clustering where each of $A$ and $B$ can be covered by a ball with radius $\leq \Delta$. So if we run the Gonzalez's algorithm $2^{2(\rho+1)}\geq 2$ rounds, we obtain a set of $2^{2(\rho+1)}$ balls with radius $\leq 2\Delta$ (since the Gonzalez's algorithm yields a $2$-approximation of $k$-center clustering). Thus the statement~(\rmnum{2}) is true.
\end{proof}

\vspace{-0.1in}
\begin{lemma}
\label{lem-hemd}
In Algorithm~\ref{alg-hemd}, for each $1\leq i\leq \log \frac{2}{\epsilon}+5$, $\mathcal{EMD}(A_i, B_i)\in \mathcal{EMD}(A, B)\pm \frac{1}{2^{i-3}}\Delta$.
\end{lemma}
\vspace{-0.1in}
\begin{proof}
First, we consider another two sets of points $\tilde{A}_i$ and $\tilde{B}_i$, where each of them contains the same set of points $\{p_{v^1_i}, p_{v^2_i}, \cdots, p_{v^N_i}\}$. To differentiate the points in $\tilde{A}_i$ and $\tilde{B}_i$, we denote each point $p_{v^j_i}$ as $a^j_i$ ({\em resp.,} $b^j_i$) in $\tilde{A}_i$ ({\em resp.,} $\tilde{B}_i$). For the set $\tilde{A}_i$, each point $a^j_i$ is associated with the weight $n^j_i$; similarly, each point $b^j_i$ of $\tilde{B}_i$ has the weight $m^j_i$. We can imagine that each $a^j_i$ is a set of $n^j_i$ unit-weight overlapping points; namely, there is a bijection between ``$a^j_i$'' and $A\cap P_{v^j_i}$. The similar bijection also exists between ``$b^j_i$'' and $B\cap P_{v^j_i}$. Moreover, since the whole set $P_{v^j_i}$ is covered by a ball with radius $\frac{\Delta}{2^{i-2}}$, through the triangle inequality, we have
\begin{eqnarray}
\mathcal{EMD}(\tilde{A}_i, \tilde{B}_i)\in \mathcal{EMD}(A, B)\pm \frac{1}{2^{i-2}}\Delta\times 2
=\mathcal{EMD}(A, B)\pm \frac{1}{2^{i-3}}\Delta. \label{for-hemd1}
\end{eqnarray}
Next, we only need to prove $\mathcal{EMD}(A_i, B_i)=\mathcal{EMD}(\tilde{A}_i, \tilde{B}_i)$. 

\vspace{-0.05in}
\begin{claim}
\label{cla-hemd}
There exists a set of flows $\tilde{F}=\{\tilde{f}_{jl}\mid 1\leq j, l\leq N\}$ yielding the optimal EMD from $\tilde{A}_i$ to $\tilde{B}_i$, such that for any $1\leq j\leq N$, $\tilde{f}_{jj}=\min\{n^j_i, m^j_i\}$.
\end{claim}
\vspace{-0.05in}

The proof of Claim~\ref{cla-hemd} is placed to our supplement. Claim~\ref{cla-hemd} indicates that the flow from $a^j_i$ to $b^j_i$ is $\min\{n^j_i, m^j_i\}$. Without loss of generality, we assume $n^j_i\leq m^j_i$; then we can safely delete the point $a^j_i$ and replace $m^j_i$ by $m^j_i-n^j_i$ without changing the value of $\mathcal{EMD}(\tilde{A}_i, \tilde{B}_i)$. If we perform this change for each pair $(a^j_i, b^j_i)$ for $1\leq j\leq N$, the point sets $\tilde{A}_i$ and $\tilde{B}_i$ will become $A_i$ and $B_i$ eventually. Therefore, $\mathcal{EMD}(A_i, B_i)=\mathcal{EMD}(\tilde{A}_i, \tilde{B}_i)$, and consequently (\ref{for-hemd1}) implies Lemma~\ref{lem-hemd} is true.
\end{proof}

\begin{proof} 
(\textbf{of Theorem~\ref{the-hemd}})
At the $i$-th level in the tree $\mathcal{H}$, we have $\mathcal{EMD}(A_i, B_i)\in \mathcal{EMD}(A, B)\pm \frac{1}{2^{i-3}}\Delta$ via Lemma~\ref{lem-hemd}. Also, since $\tilde{\Delta}/2\leq\Delta\leq \tilde{\Delta}$, we know that 
\begin{eqnarray}
\mathcal{EMD}(A, B)-\frac{1}{2^{i-3}}\tilde{\Delta}\leq \mathcal{EMD}(A_i, B_i)\leq \mathcal{EMD}(A, B)+\frac{1}{2^{i-3}}\tilde{\Delta}. \label{for-hemd2}
\end{eqnarray}

If the first event  or third event happens, the left hand-side of (\ref{for-hemd2}) implies 
$\mathcal{EMD}(A_i, B_i)>T-\frac{1}{2^{i-3}}\tilde{\Delta}$. 
So the algorithm will never output ``case 2''. Moreover, for the first event ``$\mathcal{EMD}(A, B)> T+ \epsilon \Delta$'', the bound of the height $\min\{\log\frac{1}{\epsilon}, \log\frac{\Delta}{\delta}\}+5=\log\frac{\Delta}{\delta}+5$; when $i$ reaches $\log\frac{\Delta}{\delta}+5$ 
, we have  
\begin{eqnarray}
\mathcal{EMD}(A_i, B_i)&\geq&\mathcal{EMD}(A, B)-\frac{1}{2^{i-3}}\tilde{\Delta}=T+\delta-\frac{1}{2^{i-3}}\tilde{\Delta}\label{for-the-hemd1}\\
\text{and }\hspace{0.1in}\delta&=&\frac{\Delta}{2^{i-5}}\geq\frac{\tilde{\Delta}}{2^{i-4}}.\label{for-the-hemd2}
\end{eqnarray}
The inequality of (\ref{for-the-hemd1}) comes from the left hand-side of (\ref{for-hemd2}). Combining (\ref{for-the-hemd1}) and (\ref{for-the-hemd2}), we have $\mathcal{EMD}(A_i, B_i)\geq T+\frac{1}{2^{i-3}}\tilde{\Delta}$. That is, the algorithm will output ``case 1'' before $i$ exceeds $\log\frac{\Delta}{\delta}+5$.

Similarly, for the second event ``$\mathcal{EMD}(A, B)< T- \epsilon \Delta$'', the algorithm will output ``case~2'' before $i$ exceeds $\log\frac{\Delta}{\delta}+5$.

For the third event ``$\mathcal{EMD}(A, B)\in [T, T+ \epsilon \Delta]$'', the bound of the height $\min\{\log\frac{1}{\epsilon}, \log\frac{\Delta}{\delta}\}+5=\log\frac{1}{\epsilon}+5$. The algorithm could output ``case~1''  before $i$ reaches $\log \frac{1}{\epsilon}+5$. It is also possible that the algorithm keeps running until $i=\log \frac{1}{\epsilon}+5$, and then it will output ``case~3''. Similarly, we can prove the output for the fourth event ``$\mathcal{EMD}(A, B)\in [ T- \epsilon \Delta, T]$''. 
\end{proof}
\vspace{-0.15in}

\section{Experiments}
\label{sec-exp}
All the experimental results were obtained on a server equipped with 2.4GHz Intel CPU and 8GB main memory; the algorithms are implemented in Matlab R2019a. 
As discussed in Remark~\ref{rem-the-hemd} (\rmnum{2}), we can apply any existing EMD algorithm as the sub-routine to compute $\mathcal{EMD}(A_i, B_i)$ in Step 2(d) of our  Algorithm~\ref{alg-hemd}. 
In our experiments, we use two widely used EMD algorithms, the \textbf{\textsc{Network Simplex}} algorithm~\cite{ahuja1993network} and the \textbf{\textsc{Sinkhorn}} algorithm~\cite{DBLP:conf/nips/Cuturi13}, as the sub-routine algorithms. In fact, we also considered the well-known EMD algorithm \textbf{\textsc{FastEMD}}~\cite{pele2009fast}, but it runs very slowly for high-dimensional data ({\em e.g.,} it takes several hours for computing the EMD over the datasets considered in our experiments). 
Given an instance $(A, B, T)$, we let $\mathtt{time_{our}}$ be the running time of our algorithm, and $\mathtt{time_{net}}$ ({\em resp.,} $\mathtt{time_{sin}}$) be the running time of computing $\mathcal{EMD}(A, B)$ by using \textsc{Network Simplex} ({\em resp.,} \textsc{Sinkhorn}); we use the ratios $\mathtt{time_{our}/time_{net}}$ and $\mathtt{time_{our}/time_{sin}}$ to measure the performance of our algorithm (the lower the ratio, the better the performance).

\textbf{Datasets.} We implement our proposed algorithm and study its performances on both the synthetic and real datasets as listed in Table~\ref{datasets info}. 
To construct a synthetic dataset, we take the random samples from two randomly generated manifolds in $\mathbb{R}^{500}$, where each manifold is represented by a polynomial function with low degree ($\leq 50$). 
 Note that it is challenging to achieve the exact doubling dimensions of the datasets,  so we use the degree of the polynomial function as a ``rough  indicator'' for the doubling dimension (the higher the degree, the larger the doubling dimension). We also use two popular benchmark datasets, the \textbf{MNIST} dataset~\cite{lecun1998gradient} and \textbf{CIFAR-10} dataset~\cite{krizhevsky2009learning}. 
Following Remark~\ref{rem-tree} and Remark~\ref{rem-the-hemd}(\rmnum{1}), our algorithm does not require that the doubling dimension $\rho$ is given.

\begin{table}
  \caption{The Datasets.}
  \centering
  \begin{tabular}{llll}
    \toprule
    \cmidrule(r){1-4}
    Datasets      & Data size      & Dimension      & Type    \\
    \midrule
    SYNTHETIC & $80,000$      & $500$           &  Synthetic   \\
    MNIST        & $60,000$      & $784$            &  Image      \\
    CIFAR-10    & $60,000$      & $3072$          &  Image \\
    \bottomrule
  \end{tabular}
  \label{datasets info}
\end{table}

\textbf{Setup.} We set the threshold $T=2^\theta\cdot \mathcal{EMD}(A, B)$ and vary the parameter $\theta$ from $-10$ to $10$; we set $\epsilon$, as the error parameter, to be $0.01, 0.03$, and $0.05$. For each instance $(A, B, T)$, we first use \textsc{Network Simplex} and \textsc{Sinkhorn} to compute their EMD and obtain $\mathtt{time_{net}}$ and $\mathtt{time_{sin}}$, respectively. 

%

\textbf{Results and analysis.} 
We illustrate the experimental results  in Figure~\ref{fig-syn}, ~\ref{fig-minist}, ~\ref{fig-cifar}, and ~\ref{fig-synid} .

\begin{itemize}
\item We can see that the running time ratios are lower than $0.3$ on the synthetic datasets and $0.38$ on the real datasets, which indicates that our algorithm can save respectively at least $70\%$ and $62\%$ of the running time on the synthetic and real datasets, comparing with directly computing the EMD. Furthermore, when we tune the parameter $\theta$ to be close to $0$ ({\em i.e.,} the threshold $T$ is close to $\mathcal{EMD}(A, B)$), the running time increases because the height of the tree $\mathcal{H}$ becomes high. When $T$ is far from $\mathcal{EMD}(A, B)$, the curves become flat, because $\mathcal{H}$'s height remains the same and the running time is dominated by the construction time of $\mathcal{H}$.

%


\item We show the average running time ratio and  standard deviation of $\mathtt{time_{our}/time_{net}}$ and $\mathtt{time_{our}/time_{sin}}$ in Figure~\ref{fig-syn}-\ref{fig-cifar}'s (b) and (d). When the size $n$ increases, the ratios substantially decrease, which indicates that our method enjoys better scalability for large-scale datasets. This is also in agreement with our theoretical analysis on the running time in Remark~\ref{rem-the-hemd}(\rmnum{2}).


\item We also vary the degree of the polynomial function for the synthetic datasets (with fixed $\epsilon=0.03$). From Figure~\ref{fig-synid} we can see that the running time of  small $\theta$ increases as the degree increases, because the complexity is largely affected by the value of $\rho$ when the tree $\mathcal{H}$ is high. On the other hand, when $\theta$ is large, $\mathcal{H}$ is low and the complexity is dominated by the running time for constructing $\mathcal{H}$ (so the curves become flat).


%
\end{itemize}

\begin{figure}[H]
\centering
\subfigure[$n=20,000$]{
\begin{minipage}[t]{0.23\linewidth}
\centering
\includegraphics[width=1.35in]{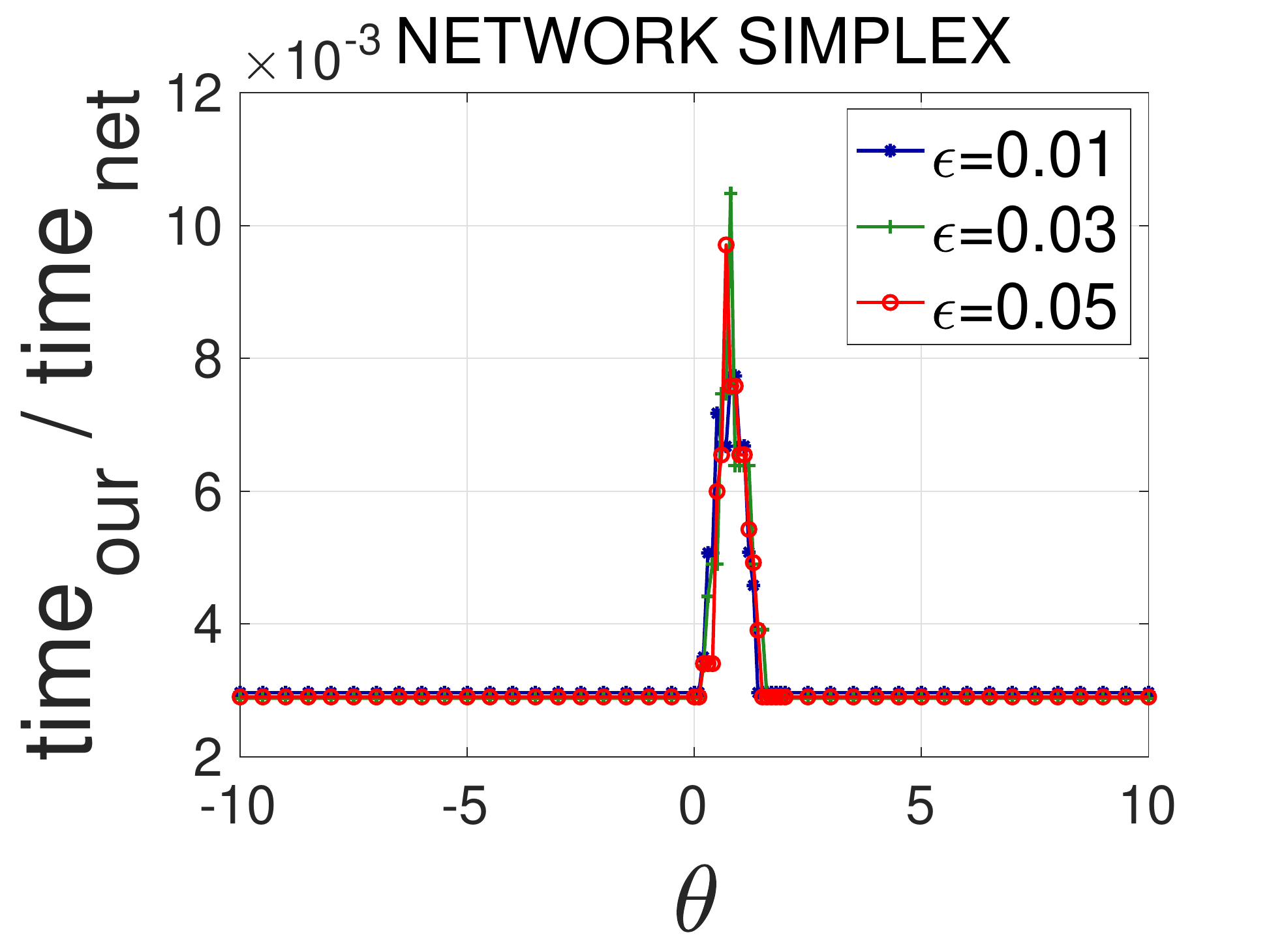}
\end{minipage}%
}%
\subfigure[]{
\begin{minipage}[t]{0.23\linewidth}
\centering
\includegraphics[width=1.35in]{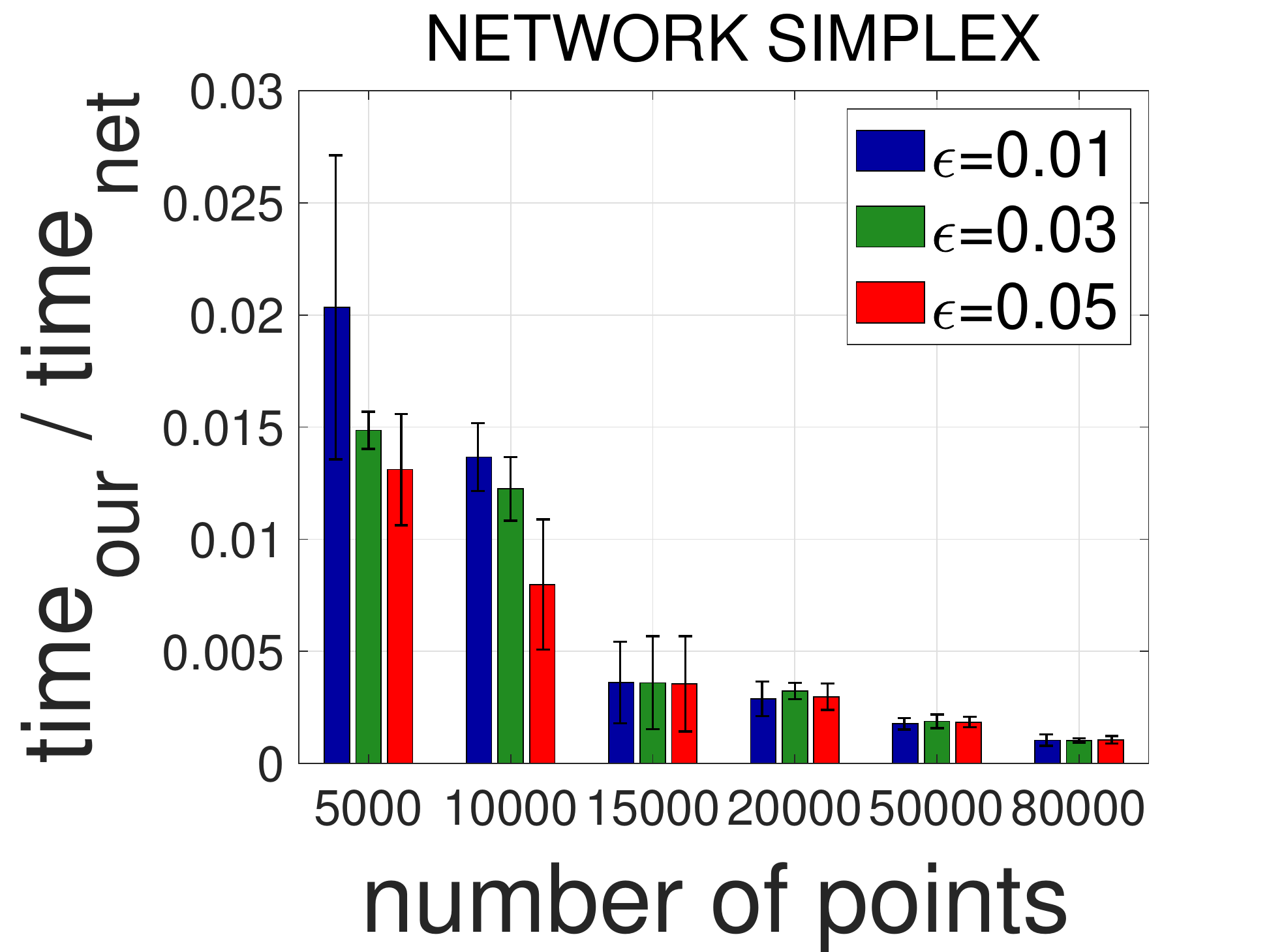}
\end{minipage}%
}%
\subfigure[$n=20,000$]{
\begin{minipage}[t]{0.23\linewidth}
\centering
\includegraphics[width=1.35in]{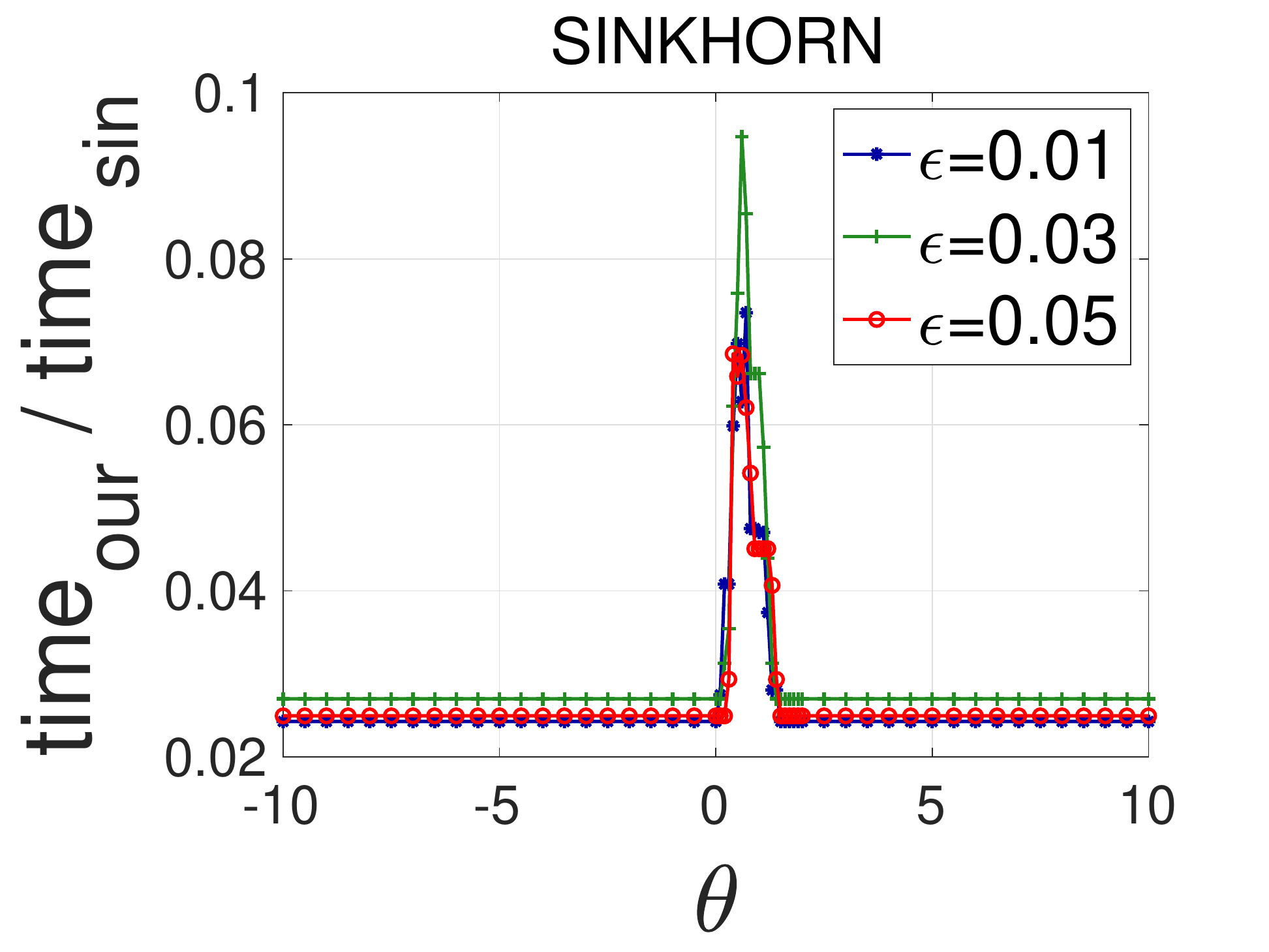}
\end{minipage}
}%
\subfigure[]{
\begin{minipage}[t]{0.23\linewidth}
\centering
\includegraphics[width=1.35in]{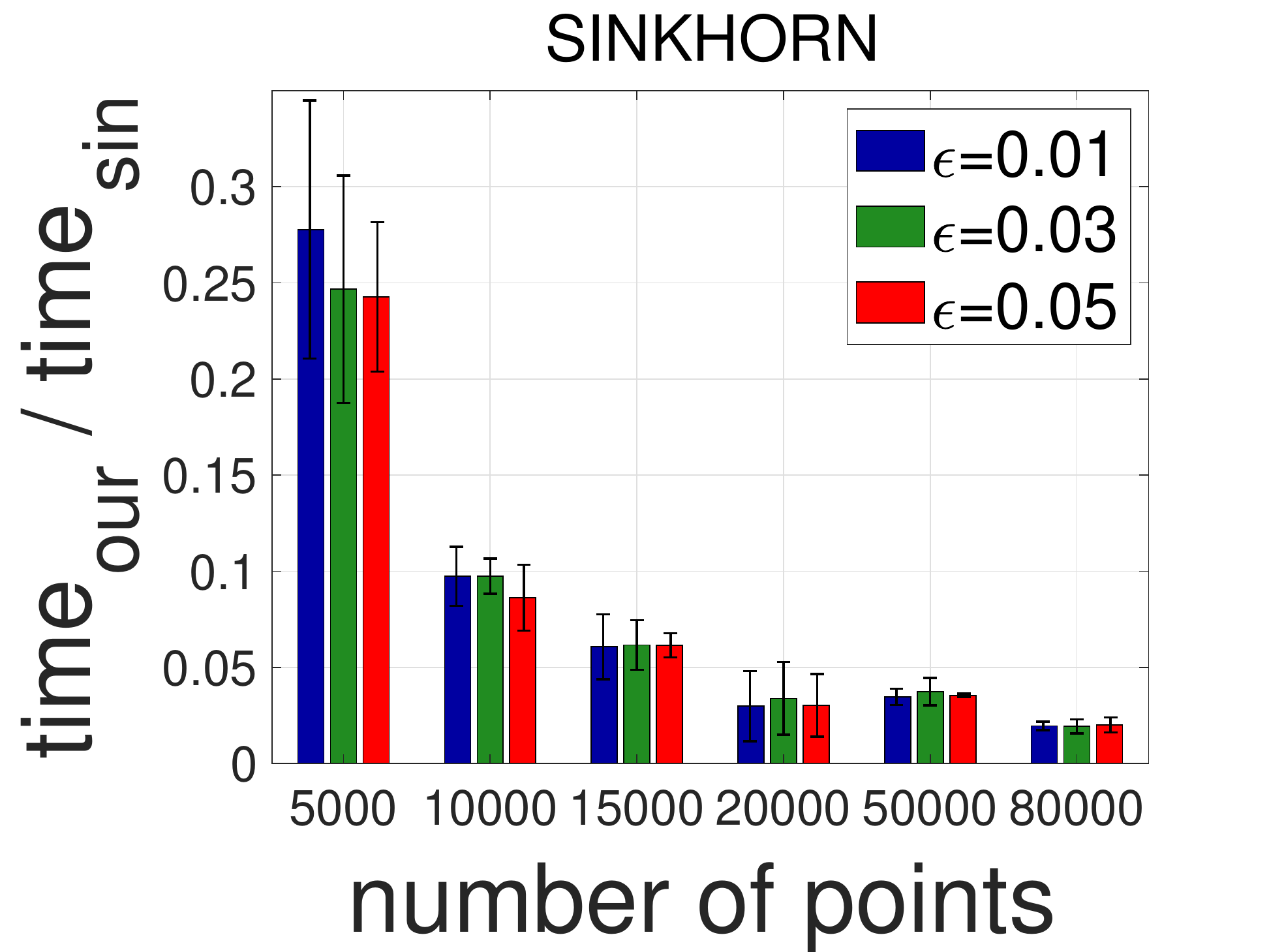}
\end{minipage}
}%
\centering
 \caption{The running time ratios on the synthetic datasets with varying the threshold $T=2^\theta\cdot \mathcal{EMD}(A, B)$ and the number of points $n=|A|+|B|$.}
  \label{fig-syn}
\end{figure}

 \begin{figure}[htb]
\centering
\subfigure[$n=20,000$]{
\begin{minipage}[t]{0.23\linewidth}
\centering
\includegraphics[width=1.35in]{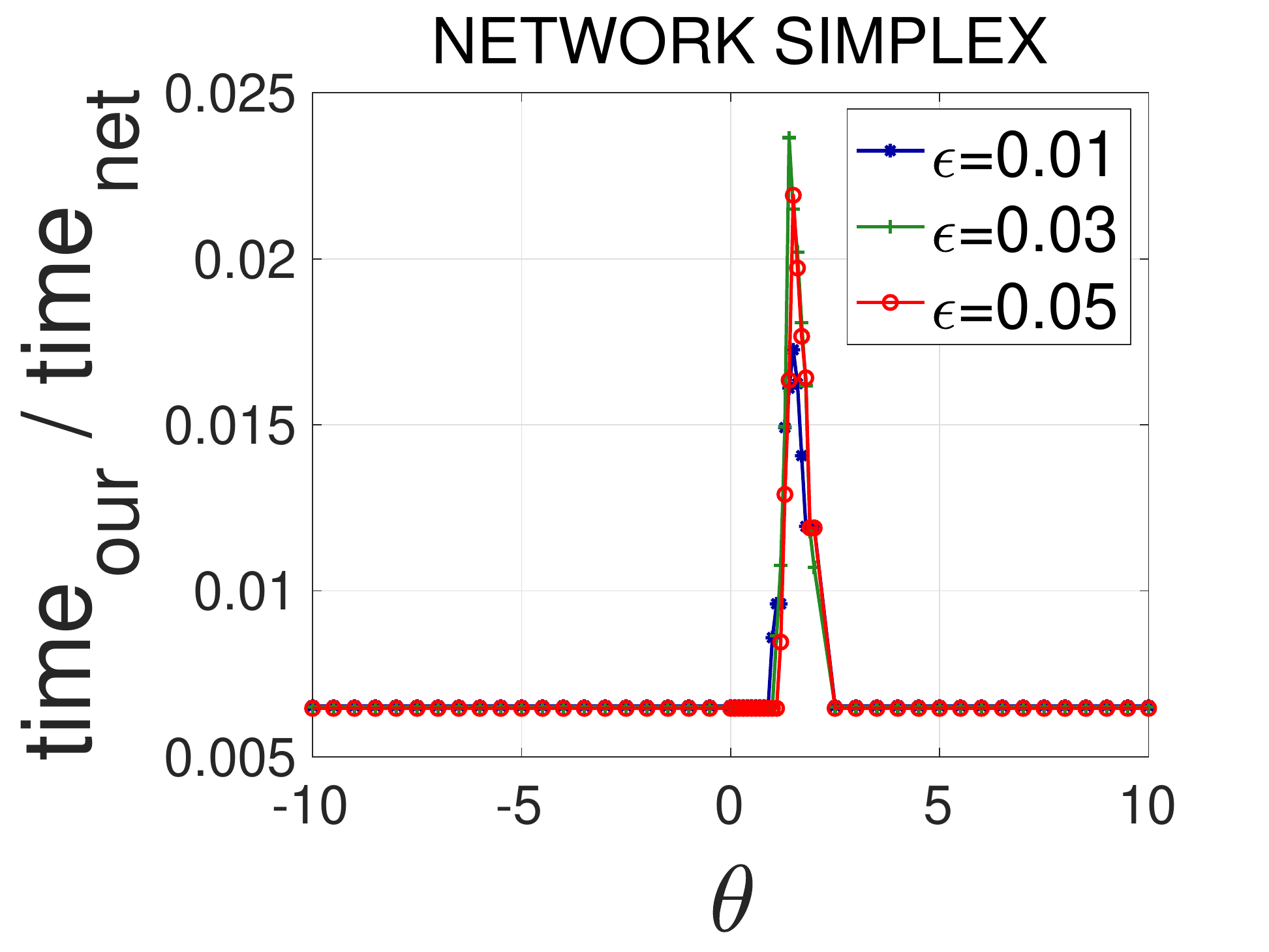}
\end{minipage}%
}%
\subfigure[]{
\begin{minipage}[t]{0.23\linewidth}
\centering
\includegraphics[width=1.35in]{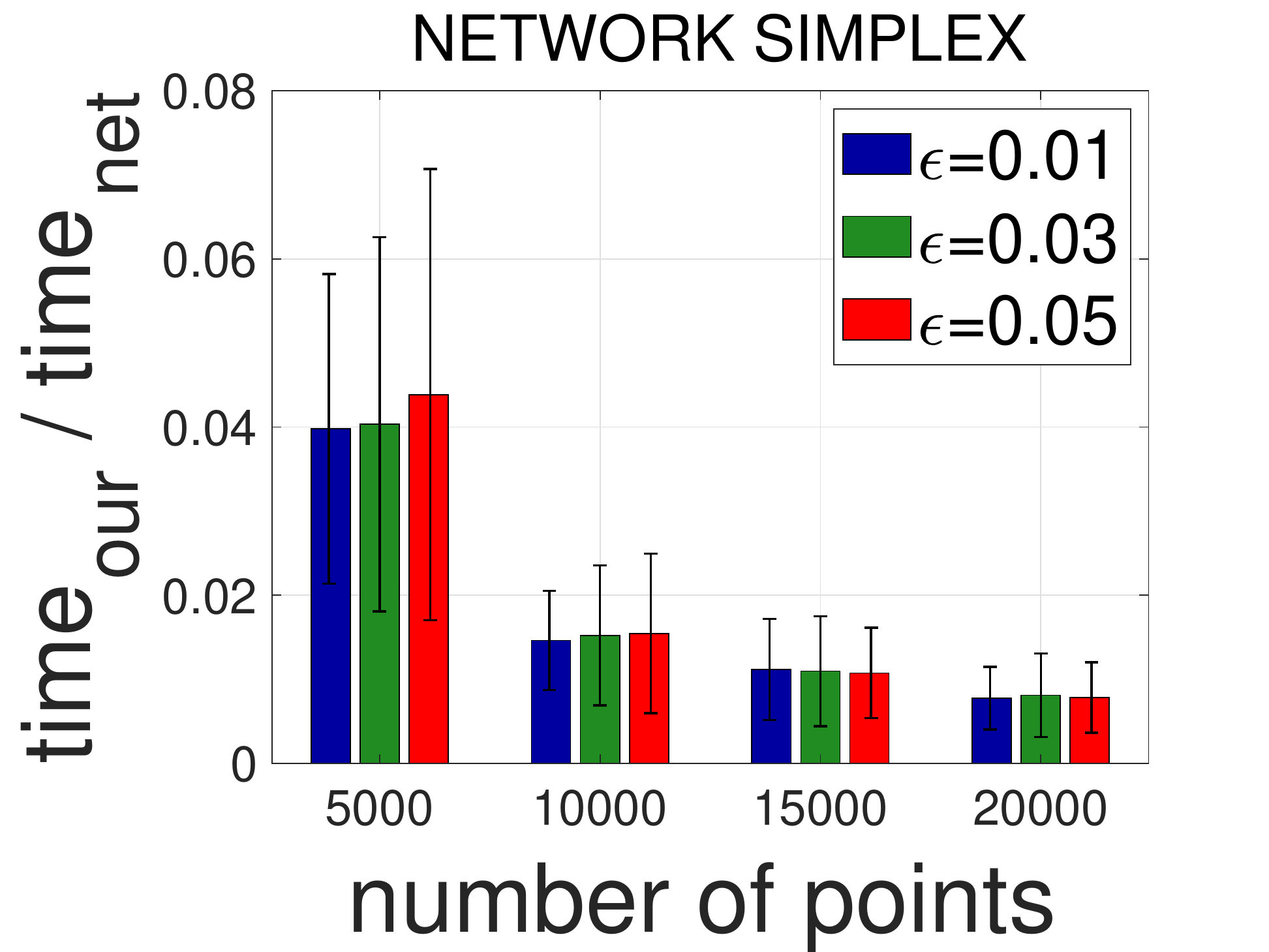}
\end{minipage}%
}%
\subfigure[$n=20,000$]{
\begin{minipage}[t]{0.23\linewidth}
\centering
\includegraphics[width=1.35in]{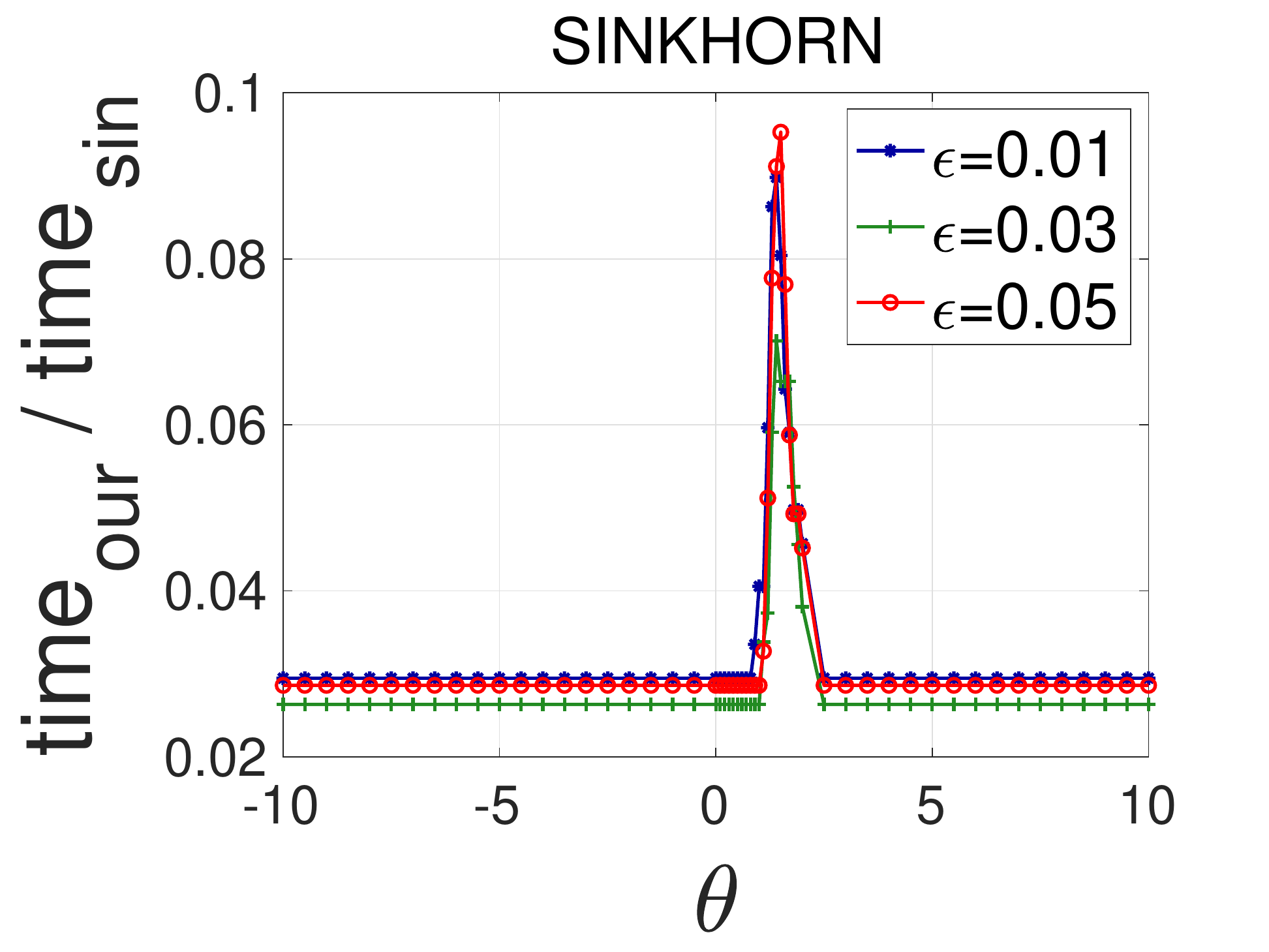}
\end{minipage}
}%
\subfigure[]{
\begin{minipage}[t]{0.23\linewidth}
\centering
\includegraphics[width=1.35in]{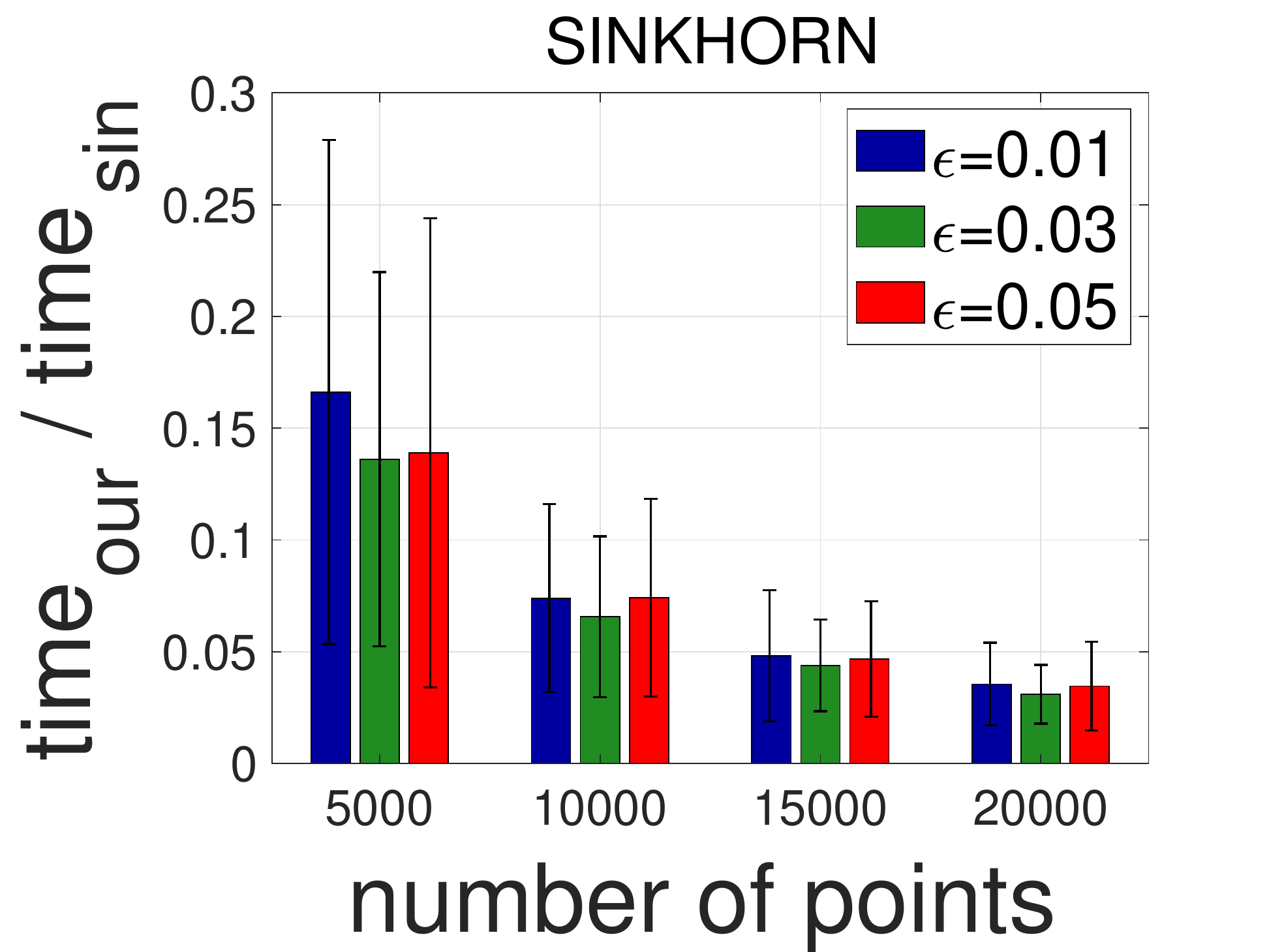}
\end{minipage}
}%
\centering
 \caption{The running time ratios on the MNIST dataset with varying the threshold $T=2^\theta\cdot \mathcal{EMD}(A, B)$ and the number of points $n=|A|+|B|$.}
  \label{fig-minist}
\end{figure}

 \begin{figure}[htb]
\centering
\subfigure[$n=20,000$]{
\begin{minipage}[t]{0.23\linewidth}
\centering
\includegraphics[width=1.35in]{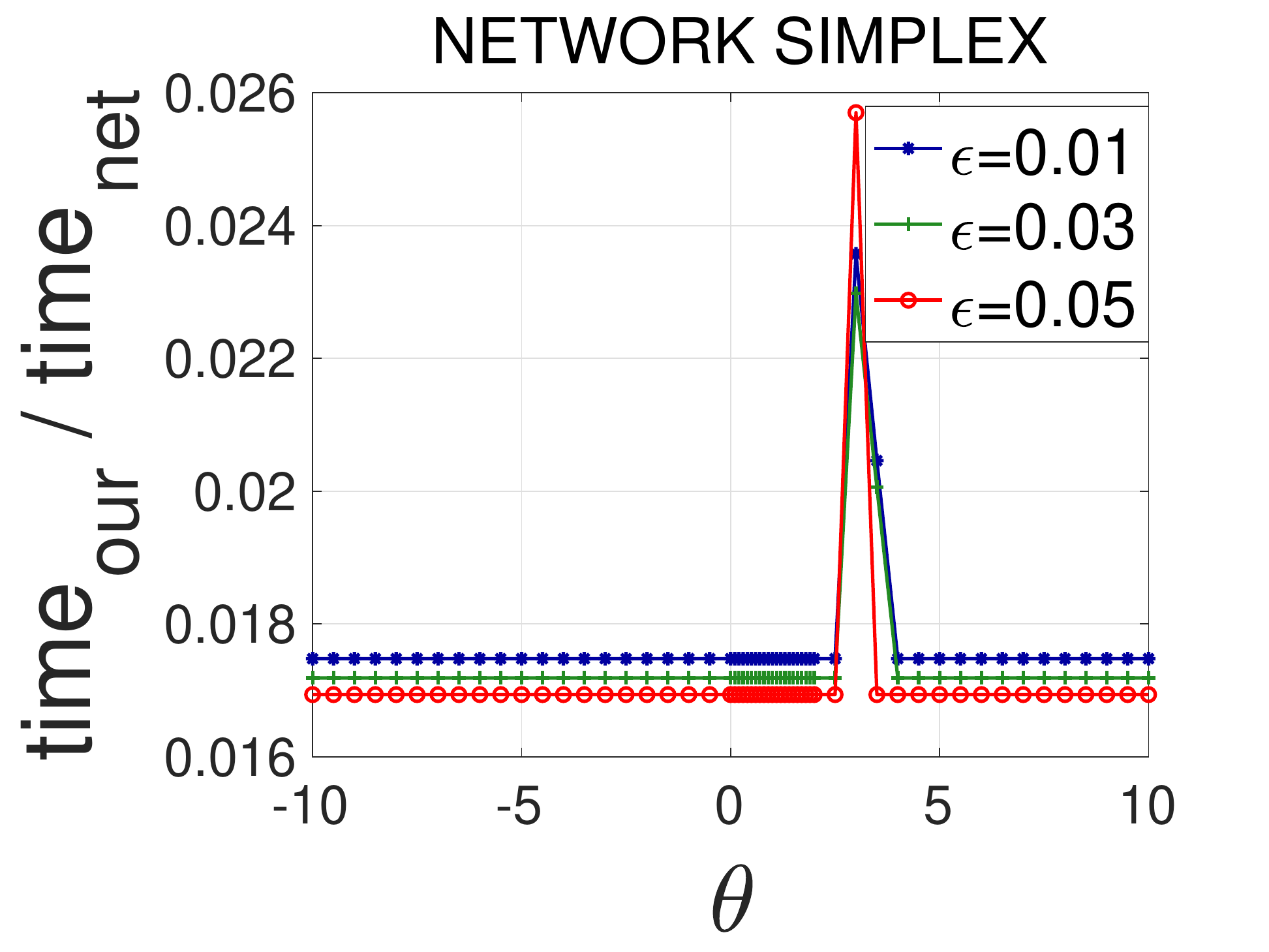}
\end{minipage}%
}%
\subfigure[]{
\begin{minipage}[t]{0.23\linewidth}
\centering
\includegraphics[width=1.35in]{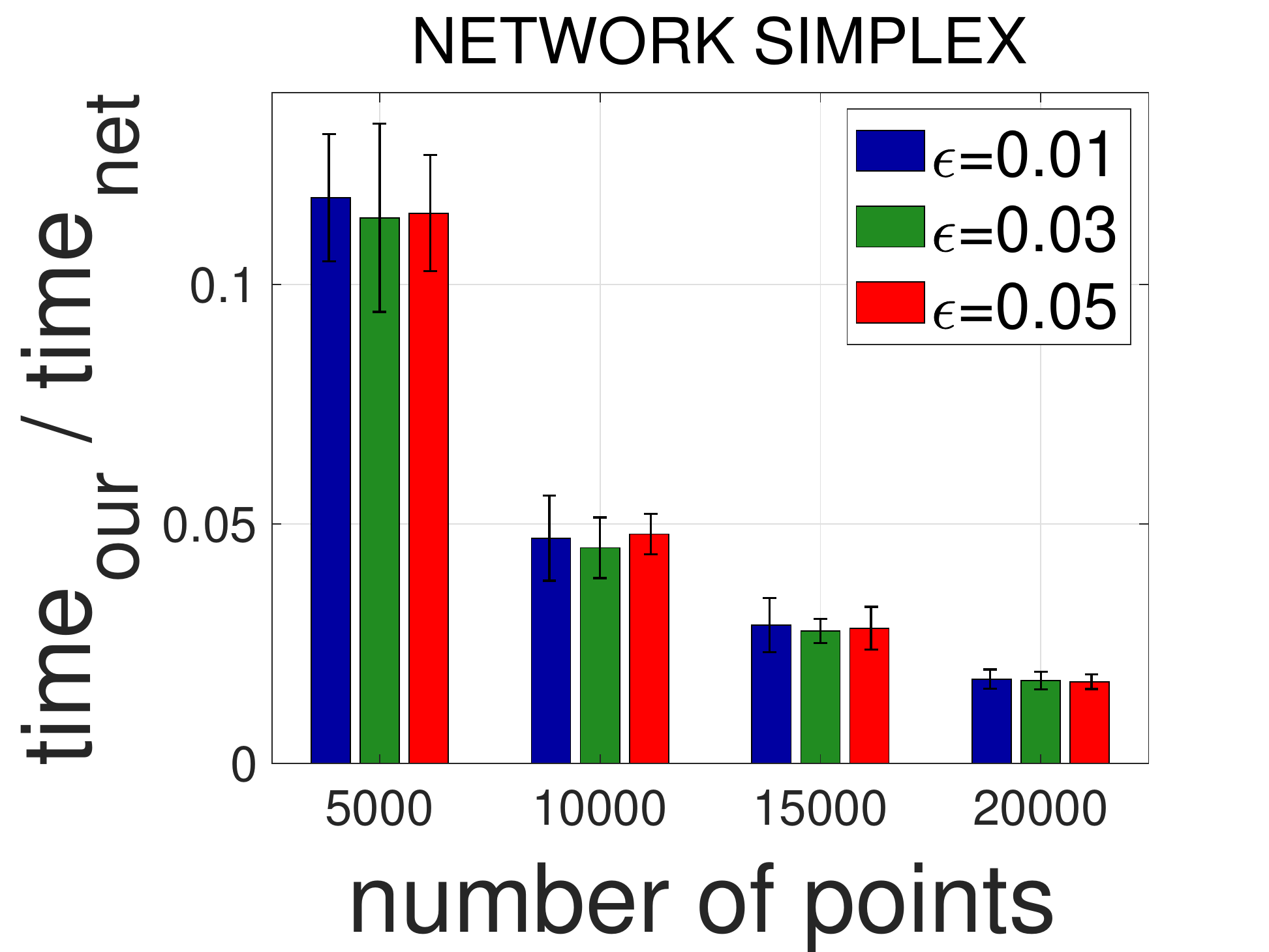}
\end{minipage}%
}%
\subfigure[$n=20,000$]{
\begin{minipage}[t]{0.23\linewidth}
\centering
\includegraphics[width=1.35in]{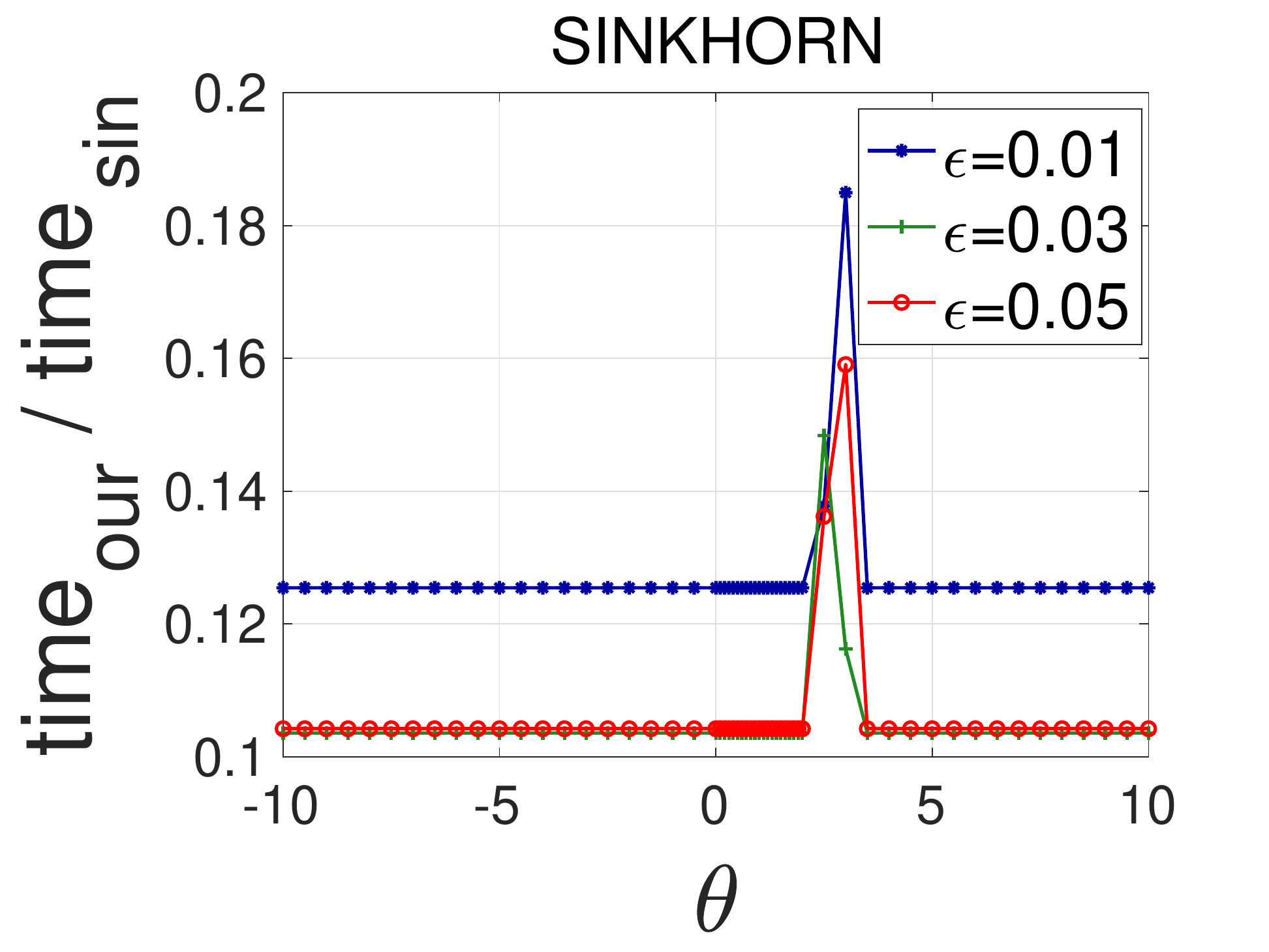}
\end{minipage}
}%
\subfigure[]{
\begin{minipage}[t]{0.23\linewidth}
\centering
\includegraphics[width=1.35in]{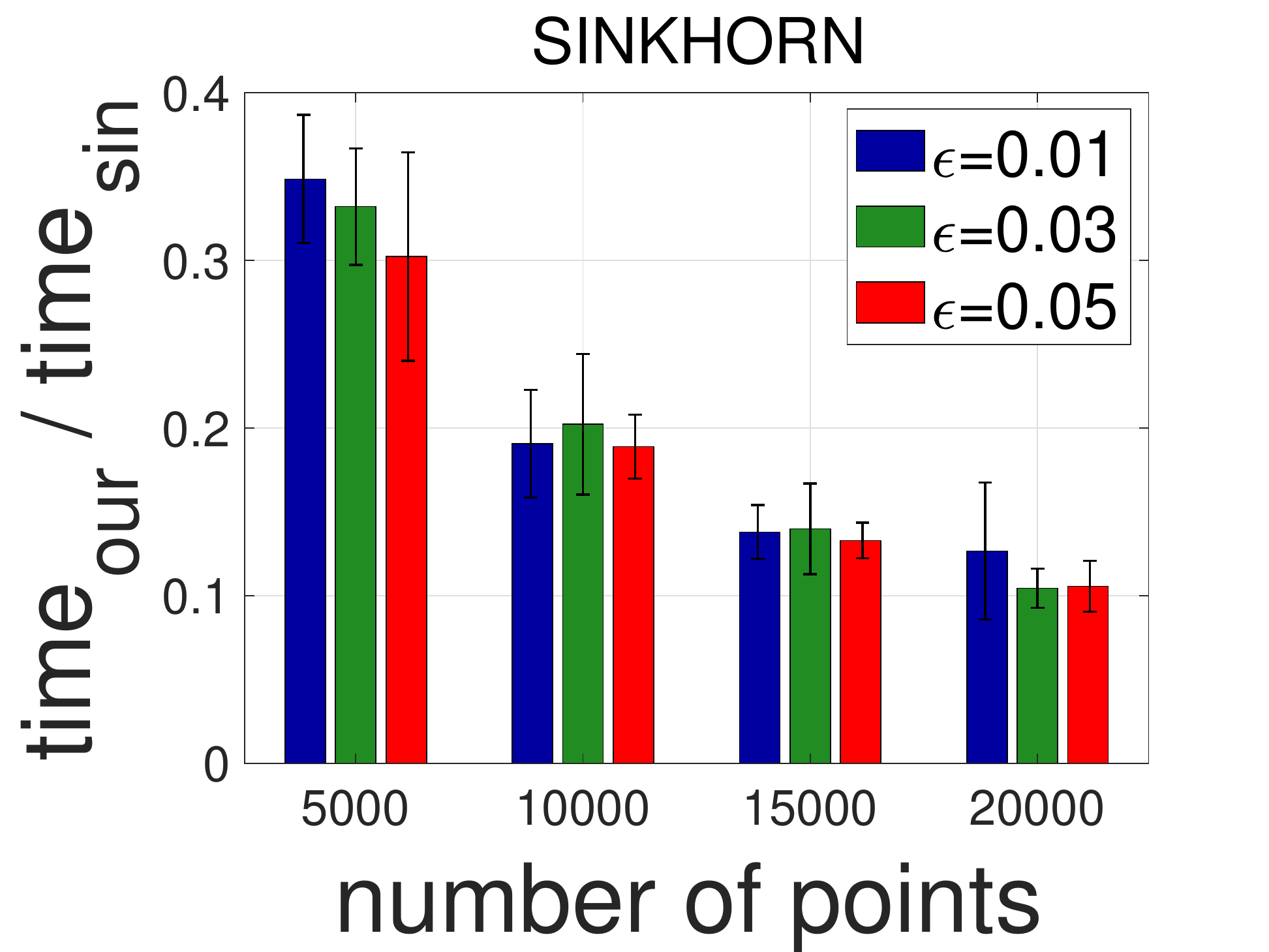}
\end{minipage}
}%
\centering
 \caption{The running time ratios on the CIFAR-10 dataset with varying the threshold $T=2^\theta\cdot \mathcal{EMD}(A, B)$ and the number of points $n=|A|+|B|$.}
  \label{fig-cifar}
\end{figure}


 \begin{figure}[htb]
\centering
\subfigure[$n=20,000$]{
\begin{minipage}[t]{0.4\linewidth}
\centering
\includegraphics[width=1.5in]{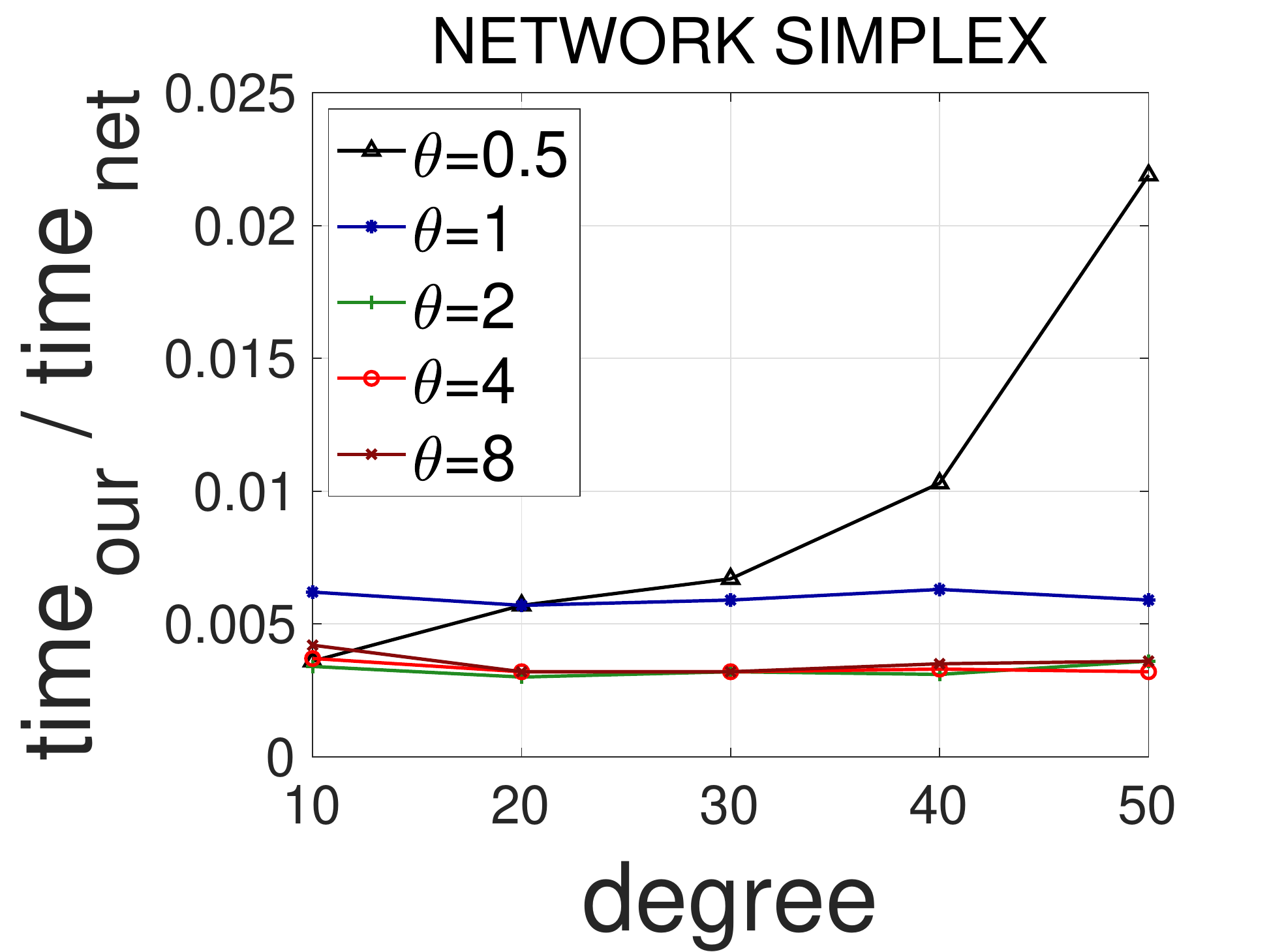}
\end{minipage}%
}%
\subfigure[$n=20,000$]{
\begin{minipage}[t]{0.4\linewidth}
\centering
\includegraphics[width=1.5in]{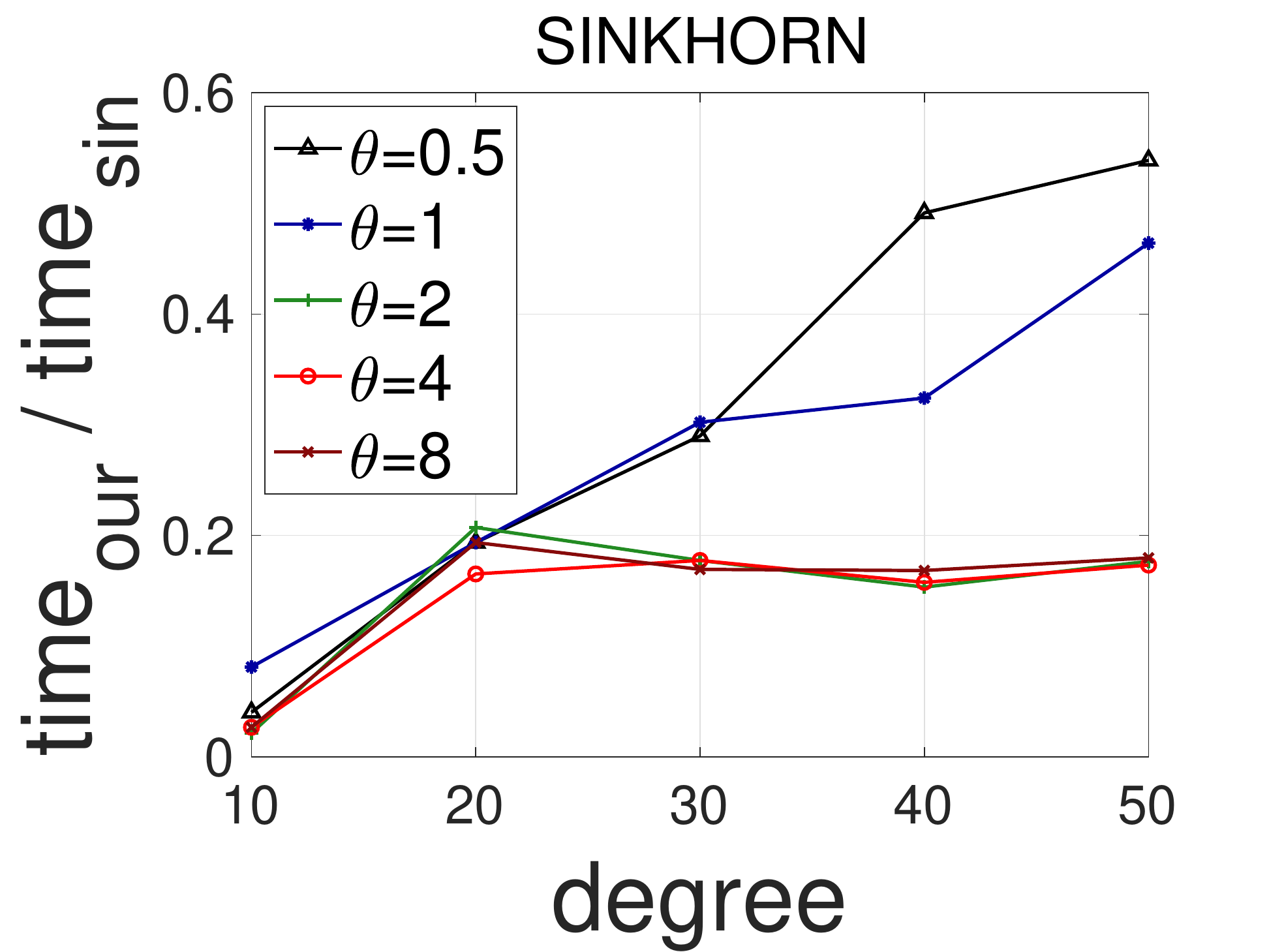}
\end{minipage}
}%
\centering
 \caption{The running time ratios on the synthetic datasets with varying the degree.}
  \label{fig-synid}
\end{figure}

We also compute the precisions of our method on the synthetic datasets and real datasets, where the precision measures the frequency that our method returns correct results with respect to the three cases defined in Section~\ref{sec-match}. Our method can achieve the average precision $\geq$ $97.6\%$ on the synthetic datasets  and $\geq 88.8\%$ on the real datasets over all the instances.

\vspace{-0.08in}
\section{Future Work}
\vspace{-0.08in}
In this paper, we propose a novel data-dependent algorithm for fast solving the EMD query problem. Our algorithm enjoys several advantages in practice. For example, it is very easy to implement and any existing EMD algorithm can be plugged as the black box to the framework. 
Following this work, it is interesting to consider generalizing our method for other measures instead of EMD ({\em e.g.,} Kullback–Leibler divergence).


%
%
%
%

\bibliographystyle{abbrv}

\bibliography{nips_2017}

\section{Proof of Claim \ref{cla-hemd}}

\begin{figure}[hbtp]
\begin{center}
    \includegraphics[width=0.49\textwidth]{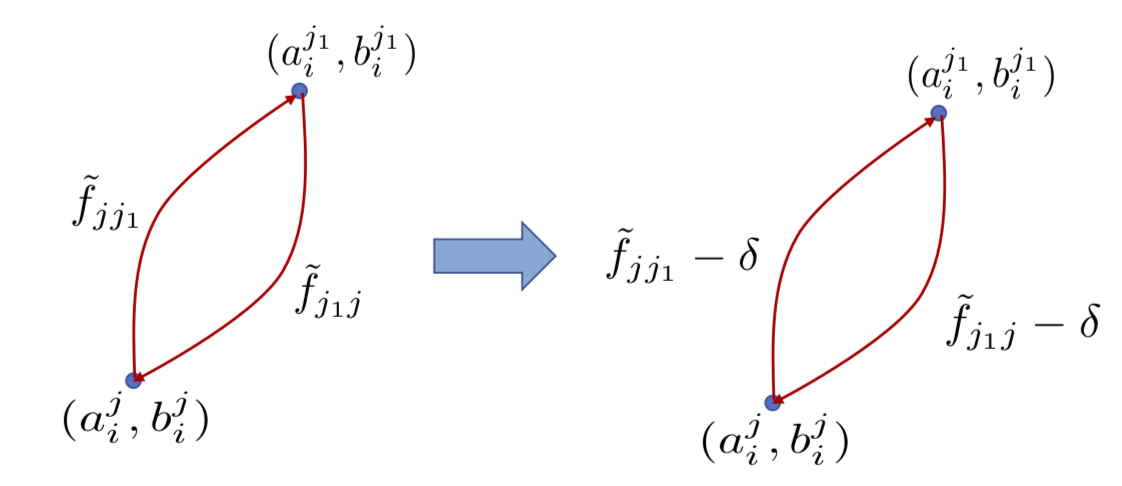}  
 %
    \includegraphics[width=0.49\textwidth]{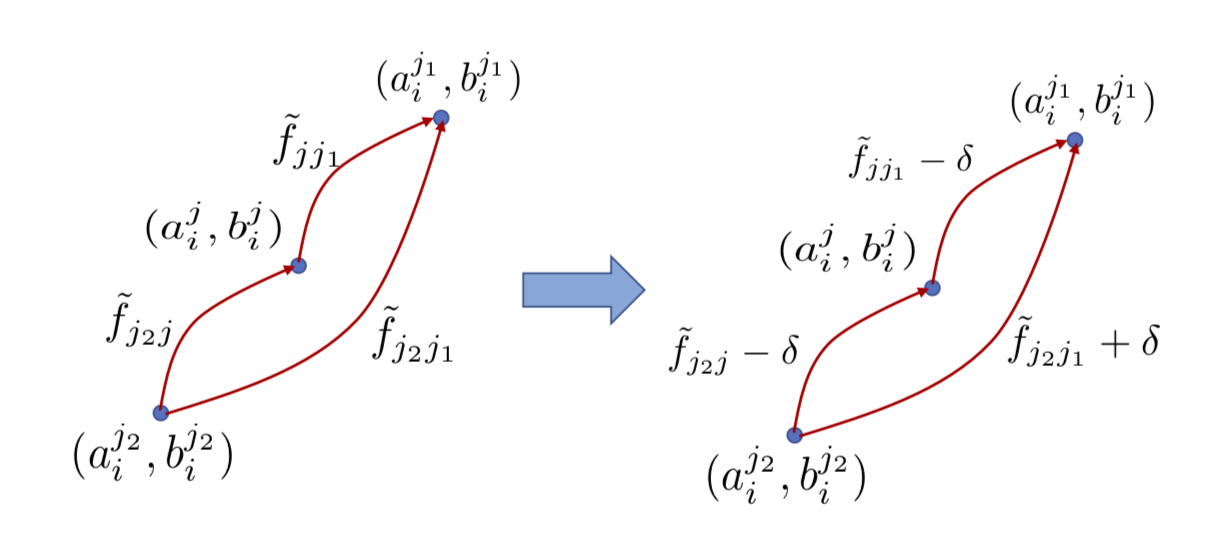}  
    \end{center}
  \caption{The illustrations for case 1 (the left figure) and  case 2 (the right figure).}     
   \label{fig-flow}
\end{figure}

Suppose that there exists a flow $\tilde{f}_{jj}\neq\min\{n^j_i, m^j_i\}$, {\em i.e.,} $\tilde{f}_{jj}<\min\{n^j_i, m^j_i\}$ (note that $\tilde{f}_{jj}$ cannot be larger than $\min\{n^j_i, m^j_i\}$). 
Then, there must exist another flow $\tilde{f}_{jj_1}>0$ from the point $a^j_i$. Then, we consider two cases: (1) $\tilde{f}_{j_1 j}>0$ and (2) $\tilde{f}_{j_1 j}=0$, where $\tilde{f}_{j_1 j}$ is the flow from $a^{j_1}_i$ to $b^j_i$. 

For case (1), let $\delta=\min\{\tilde{f}_{jj_1}, \tilde{f}_{j_1 j}\}$. We replace the flows $\tilde{f}_{jj}, \tilde{f}_{jj_1}$, $\tilde{f}_{j_1 j}$, and $\tilde{f}_{j_1 j_1}$, by $\tilde{f}_{jj}+\delta$, $\tilde{f}_{j j_1}-\delta$, $\tilde{f}_{j_1 j}-\delta$, and $\tilde{f}_{j_1 j_1}+\delta$, respectively (see the left figure in Figure~\ref{fig-flow}). It is easy to know that the new flows are still feasible and the matching cost is reduced by $2\delta ||p_{v^j_i}-p_{v^{j_1}_i}||\geq 0$. Also, either $\tilde{f}_{j j_1}-\delta$ or $\tilde{f}_{j_1 j}-\delta$ is equal to $0$.

For case (2), since $\tilde{f}_{j_1 j}=0$ and $\tilde{f}_{jj}<\min\{n^j_i, m^j_i\}$, there must exist a flow $\tilde{f}_{j_2 j}>0$ with $j_2\neq j_1$ (otherwise, the total flow received in $b^j_i$ is only $\tilde{f}_{jj}$ that is less than $m^j_i$). Let $\delta=\min\{\tilde{f}_{jj_1}, \tilde{f}_{j_2 j}\}$. We replace the flows $\tilde{f}_{jj}, \tilde{f}_{jj_1}$, $\tilde{f}_{j_2 j}$, and $\tilde{f}_{j_2 j_1}$, by $\tilde{f}_{jj}+\delta$, $\tilde{f}_{jj_1}-\delta$, $\tilde{f}_{j_2 j}-\delta$, and $\tilde{f}_{j_2 j_1}+\delta$, respectively (see the right figure in Figure~\ref{fig-flow}). It is easy to know that the new flows are still feasible and the matching cost is reduced by
\begin{eqnarray}
\delta (||p_{v^{j_1}_i}-p_{v^j_i}||+||p_{v^{j_2}_i}-p_{v^j_i}||-||p_{v^{j_1}_i}-p_{v^{j_2}_i}||)\geq 0
\end{eqnarray}
via the triangle inequality. Also, either $\tilde{f}_{j j_1}-\delta$ or $\tilde{f}_{j_2 j}-\delta$ is equal to $0$.

Overall, for both cases, we can always transform some non-zero flow to be $0$ without increasing the matching cost. Thus, after a finite number of steps, there should be no $\tilde{f}_{jj}\neq\min\{n^j_i, m^j_i\}$; that is,  $\tilde{f}_{jj}$ is equal to $\min\{n^j_i, m^j_i\}$ for any $1\leq j\leq N$.

\end{document}